\documentclass[a4paper,twocolumn,notitlepage,nofootinbib,longbibliography,superscriptaddress,floatfix]{revtex4-2}
\usepackage{
    adjustbox,
    algorithm,
    algpseudocode,
    amsmath,
    amssymb,
    amsthm,
    booktabs,
    braket,
    comment,
    csquotes,
    enumitem,
    graphicx,
    IEEEtrantools,
    mathtools,
    multirow,
    natbib,
    physics,
    refcount,
    relsize,
    url,
    times,
    xcolor
}

\usepackage[english]{babel}
\usepackage[normalem]{ulem}
\usepackage[caption=false]{subfig}

\definecolor{mainblue}{HTML}{1f77b4}
\definecolor{mainorange}{HTML}{ff7f0e}
\definecolor{maingreen}{HTML}{2ca02c}
\definecolor{mainred}{HTML}{DC3522}
\definecolor{mainpurple}{HTML}{9467bd}
\definecolor{mainpink}{HTML}{e377c2}

\renewcommand{\exp}{\ensuremath{\mathrm{exp}}}

\newcommand{\C}{\mathbb C}
\newcommand{\R}{\mathbb R}
\newcommand{\N}{\mathbb N}
\newcommand{\E}{\mathbb E}
\newcommand{\Prob}{\mathbb P}

\newcommand{\calL}{\mathcal L}
\newcommand{\calD}{\mathcal D}
\newcommand{\calY}{\mathcal Y}

\newcommand{\id}{\operatorname{id}}
\newcommand{\vol}{\operatorname{vol}}
\newcommand{\re}{\operatorname{Re}}

\newcommand{\din}{d_{\mathrm{in}}}
\newcommand{\dout}{d_{\mathrm{out}}}
\newcommand{\D}{D}

\newcommand{\proj}[1]{|#1\rangle\!\langle #1|}

\newcommand{\normone}[1]{\left\|#1\right\|_1}
\newcommand{\normtwo}[1]{\left\|#1\right\|_2}
\newcommand{\norminf}[1]{\left\|#1\right\|_\infty}
\newcommand{\normdiam}[1]{\left\|#1\right\|_\diamond}

\newcommand{\cgin}{c_{\mathrm{gin}}}
\newcommand{\cvol}{c_{\mathrm{vol}}}
\newcommand{\Cdens}{C_{\mathrm{dens}}}
\newcommand{\Ktilt}{K_{\mathrm{tilt}}}

\newcommand{\cac}{c_{\mathrm{ac}}}

\providecommand{\to}{\ensuremath{\Tilde{o}}}

\providecommand{\calD}{\ensuremath{\mathcal{D}}}

\providecommand{\calL}{\ensuremath{\mathcal{L}}}

\providecommand{\calY}{\ensuremath{\mathcal{Y}}}

\usepackage[dvipsnames]{xcolor}
\usepackage[colorlinks=true, urlcolor=RoyalBlue, linkcolor=RoyalBlue, citecolor=ForestGreen]{hyperref}

\DeclareMathOperator*{\argmin}{arg\,min}

\newtheorem{theorem}{Theorem}
\newtheorem{lemma}{Lemma}
\newtheorem{proposition}{Proposition}

\newtheorem{corollary}{Corollary}

\newtheorem{innercustomgeneric}{\customgenericname}
\providecommand{\customgenericname}{}
\newcommand{\newcustomtheorem}[2]{%
  \newenvironment{#1}[1]
  {%
   \renewcommand\customgenericname{#2}%
   \renewcommand\theinnercustomgeneric{##1}%
   \innercustomgeneric
  }
  {\endinnercustomgeneric}
}

\newcustomtheorem{customtheorem}{Theorem}
\newcustomtheorem{customlemma}{Lemma}
\newcustomtheorem{customproposition}{Proposition}

\usepackage[nameinlink,noabbrev]{cleveref}

\newcommand{\fu}{Dahlem Center for Complex Quantum Systems, Freie Universit\"{a}t Berlin, 14195 Berlin, Germany}

\newcommand{\har}{School of Engineering and Applied Sciences, Harvard University, Allston, MA 02134, USA}

\begin{document}

\title{Quantum memory advantage for quantum process tomography}

\author{Carlos Bravo-Prieto}
%\thanks{Corresponding author: \href{mailto:c.bravo.prieto@fu-berlin.de}{c.bravo.prieto@fu-berlin.de}}
\email{c.bravo.prieto@fu-berlin.de}
\affiliation{\fu}

\author{Weiyuan Gong}
\email{wgong@g.harvard.edu}
\affiliation{\har}

\author{Antonio Anna Mele}
\email{a.mele@fu-berlin.de}
\affiliation{\fu}

\begin{abstract}
Quantum process tomography, the task of learning an unknown quantum channel from black-box access, is a central problem in quantum information. In this setting, protocols with quantum memory can coherently store and jointly process quantum information obtained from multiple channel uses, whereas protocols without quantum memory must measure after each use and retain only a classical transcript of the measurement outcomes. A fundamental open question is whether quantum memory provides a query-complexity advantage even when protocols without quantum memory may adapt their experiments based on all previous outcomes with unbounded classical computational power.
In this work, we show that it does. We determine the optimal query complexity of quantum process tomography without quantum memory up to a \emph{constant} factor to be
$\Theta(d_{\mathrm{in}}^3 d_{\mathrm{out}}^3/\varepsilon^2)$,
where $d_{\mathrm{in}}$ and $d_{\mathrm{out}}$ are the channel input and output dimensions, respectively, and $\varepsilon$ is the target diamond-norm accuracy. More precisely, we prove that any incoherent protocol for this task, including adaptive protocols, requires
$\Omega(d_{\mathrm{in}}^3 d_{\mathrm{out}}^3/\varepsilon^2)$
queries, even when each channel use may be assisted by arbitrary fresh ancilla, and we present a non-adaptive, ancilla-free incoherent protocol achieving the matching upper bound
$O(d_{\mathrm{in}}^3 d_{\mathrm{out}}^3/\varepsilon^2)$.
Our results thereby generalize the optimal sample-complexity bounds for single-copy state tomography, recovered as the special case $d_{\mathrm{in}}=1$.
By contrast, coherent protocols with quantum memory achieve query complexity
$\Theta(d_{\mathrm{in}}^2 d_{\mathrm{out}}^2/\varepsilon^2)$.
Hence, our results establish a rigorous learning separation between quantum process tomography with and without quantum memory.
\end{abstract}

\maketitle

\section{Introduction} \label{s:introduction}

Quantum process tomography is the task of learning an unknown quantum channel from query access, which serves as the key subroutine for calibration, verification, characterization, and certification across various quantum platforms~\cite{chuang1997prescription,poyatos1997complete,helsen2022general,mohseni2008quantum,obrien2004quantum,riebe2006process,bialczak2010quantum,ballance2016high,bouchard2019quantum,blume2017demonstration,scott2008optimizing,kitaev1997quantum,gilchrist2005distance,watrous2012simpler,Watrous2018,aharonov1998quantum,hayashi1998asymptotic,odonnell2016efficient,mele2025optimal}.
The problem is formulated as: given queries to an unknown underlying quantum channel~\cite{Choi1975}, output a classical description of a channel that is close to the actual channel in diamond norm, which captures the worst-case distinguishability of two processes where the user of the process may choose an arbitrary input state, attach an arbitrary reference system, and perform an arbitrary final measurement~\cite{kitaev1997quantum,gilchrist2005distance,Watrous2018,watrous2012simpler,aharonov1998quantum}.

Beyond the characterization of quantum hardware, learning an unknown process is the natural task whenever the object of interest is a physical transformation rather than a state.
In quantum communication and quantum networks, the channel connecting distant parties must be estimated in order to adapt encoding and decoding to it and to certify the quality of the link~\cite{bouchard2019quantum,wehner2018quantum}.
In quantum sensing and metrology, the quantity to be measured is associated to a probe by a dynamical process, so that estimation amounts to learning parameters of an unknown channel~\cite{degen2017quantum,escher2011general}.
In physics and chemistry, process tomography reconstructs the effective dynamics of open quantum systems, such as energy transfer in molecular complexes~\cite{yuenzhou2011quantum}, and its multi-time extensions capture non-Markovian memory effects~\cite{pollock2018non,white2020demonstration}.
In quantum learning theory, learning an unknown process from experimental data is a basic primitive~\cite{caro2024learning,huang2023learning,raza2024online}, both for classical learners and for quantum-enhanced experiments in which quantum data are processed by a quantum computer~\cite{aharonov2022quantum,Huang_2022}.
In all of these settings, the number of uses of the unknown process is the central cost, and it is therefore important to determine which capabilities of the learner can reduce it.

An essential task in this problem is to understand the role of quantum memory.
A coherent process tomography protocol can preserve quantum systems across different uses of the unknown channel and is able to finally perform a joint quantum operation or measurement on all retained systems. 
An incoherent protocol, by contrast, uses the unknown channel one time at a time, measures immediately after each use, and stores only a classical transcript between different rounds, see Figure~\ref{fig:scheme}. 
The incoherent protocol may still be adaptive: after observing the previous outcomes, it may choose a new input state, a new measurement, and even a fresh ancilla for the next channel use. 
This distinction is already known to be essential and is well-understood in quantum state tomography.
For learning an arbitrary $d$-dimensional quantum state in trace distance, coherent collective measurements achieve the optimal sample complexity $\Theta(d^2)$  
while any incoherent single-copy measurement protocol requires $\Theta(d^3)$ even when the measurements are chosen adaptively~\cite{haah2016sample,odonnell2016efficient,kueng2017low,chen2023does}.
This result separated two resources that are often conflated: adaptivity, the ability to choose later experiments based on earlier classical data, and coherence, the ability to keep quantum systems coherent across different copies.
A more recent result even provides a smooth tradeoff between sample complexity and the number of copies on which each measurement is able to perform jointly~\cite{chen2024tradeoff}.
For structured channels like Pauli channels, it is shown that entanglement and quantum memory can also lead to exponential separations in the number of queries~\cite{aharonov2022quantum,chen2022exponential,Huang_2022,chen2022quantum,fawzi2025lower,chen2025efficient,chen2024tight}.

\begin{figure*}[t]
\centering
\includegraphics[width=\textwidth]{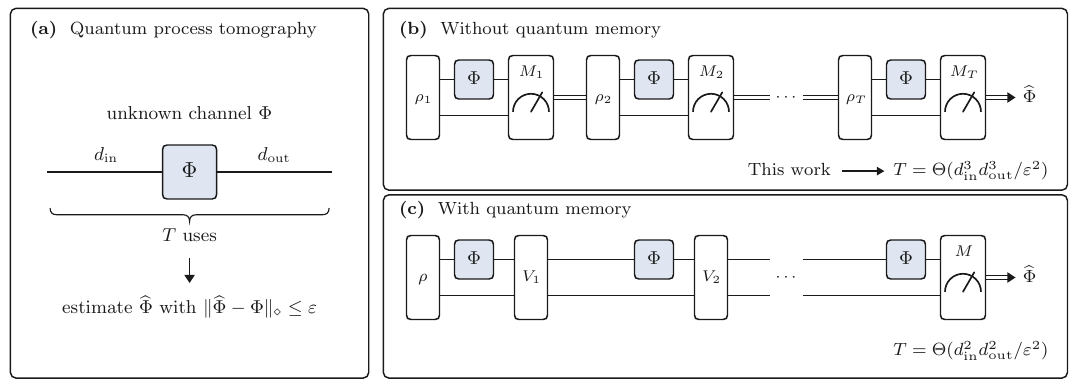}
\caption{\textbf{Quantum process tomography with and without quantum memory.}
\textbf{(a)}~An unknown channel $\Phi$ with input dimension $\din$ and output dimension $\dout$ is used $T$ times, and the goal is a classical estimate $\widehat\Phi$ with $\|\widehat\Phi-\Phi\|_\diamond\le\varepsilon$.
\textbf{(b)}~Without quantum memory: each use of $\Phi$ is followed by a measurement $M_k$, possibly on a fresh ancilla, and only classical data are carried to the next round, where $\rho_{k+1}$ and $M_{k+1}$ may be chosen adaptively.
\textbf{(c)}~With quantum memory: the outputs of all uses are kept coherent, interleaved with operations $V_k$, and measured jointly at the end.
Quantum memory reduces the optimal number of uses from $\Theta(\din^3\dout^3/\varepsilon^2)$ to $\Theta(\din^2\dout^2/\varepsilon^2)$.}
\label{fig:scheme}
\end{figure*}

However, the analogous question for quantum processes learning has remained more subtle for general channels with $d_{\mathrm{in}}$ input and $d_{\mathrm{out}}$ output dimensions. 
A general channel has a normalized Choi operator on a Hilbert space of $d_{\mathrm{in}}d_{\mathrm{out}}$ dimensions~\cite{Choi1975}. 
However, the operational error in channel tomography is the diamond norm rather than the Choi trace norm in state tomography, which prevents a direct transfer of state-tomography bounds.
Previous projected least-squares process tomography protocols gave rigorous incoherent upper bounds of $\widetilde{O}(d_{\mathrm{in}}^3d_{\mathrm{out}}^3/\varepsilon^2)$ for achieving $\varepsilon$-accuracy in diamond norm, which includes logarithmic factors~\cite{SurawyStepneyKahnKuengGuta2022,oufkir2023sample}, and proved a matching lower bound $\Omega(d_{\mathrm{in}}^3d_{\mathrm{out}}^3/\varepsilon^2)$ only for non-adaptive incoherent measurements~\cite{oufkir2023sample}. 
On the other hand, recent coherent protocols enable full process tomography in diamond norm with $\Theta(d_{\mathrm{in}}^2d_{\mathrm{out}}^2/\varepsilon^2)$ queries~\cite{mele2025optimal,chen2026quantum}.
These coherent protocols establish the best possible scaling when quantum memory is allowed, leaving open whether classical adaptivity could reduce the sample complexity without quantum memory.

This work resolves this question by giving a negative answer.
We prove that every adaptive incoherent protocol for learning an arbitrary channel with $d_{\mathrm{in}}$ input and $d_{\mathrm{out}}$ output dimensions up to $\varepsilon$ accuracy in diamond norm requires $\Omega(d_{\mathrm{in}}^3d_{\mathrm{out}}^3/\varepsilon^2)$ queries.
The lower bound holds for any incoherent protocol with arbitrary classical adaptivity, arbitrary outcome-dependent measurements, and interaction with arbitrary fresh ancilla systems within each round. 
Our results also generalize the optimal sample-complexity bounds for single-copy state tomography: when $d_{\mathrm{in}}=1$, channels reduce to states, and our lower and upper bounds recover the corresponding
$\Theta(d_{\mathrm{out}}^3/\varepsilon^2)$
scaling~\cite{chen2023does,kuengfasttkomogr}.

The proof combines channel-specific ingredients with the posterior-tilt strategy developed for adaptive incoherent state tomography~\cite{chen2023does}, in which likelihood control and local volume growth yield posterior anti-concentration.
First, we construct a local family of channels around the completely depolarizing channel by slightly perturbing the normalized Choi operator in off-diagonal blocks while enforcing trace preservation exactly.
We represent each individual query in an adaptive incoherent protocol by a one-slot tester acting on the Choi operator using the standard tester and quantum-comb formalism~\cite{GutoskiWatrous2007,ChiribellaDArianoPerinotti2009}.
After conditioning on a full transcript, the adaptive choices of the protocol become a deterministic sequence of tester elements.
Second, for every such fixed transcript, we prove a likelihood tilt bound over a local Schatten neighborhood of the true channel.
The channel-specific estimate uses the symmetry of the local perturbation neighborhood, together with positivity, to control arbitrary positive tester elements.
This bound shows that the posterior cannot concentrate inside the Choi trace ball that would be forced by a diamond-norm accurate estimator with $o(d_{\mathrm{in}}^3d_{\mathrm{out}}^3/\varepsilon^2)$ queries.
The resulting posterior anti-concentration contradicts uniform successful tomography and yields the claimed diamond-norm lower bound even with arbitrary fresh ancillas.

We complement the lower bound with a non-adaptive, ancilla-free incoherent protocol using $O(d_{\mathrm{in}}^3d_{\mathrm{out}}^3/\varepsilon^2)$ queries, which removes the logarithmic factor from the previously best incoherent process tomography guarantees~\cite{SurawyStepneyKahnKuengGuta2022,oufkir2023sample}.
Importantly, the protocol we consider remains the same as the previous work, while we carried out a refined concentration analysis in the same spirit of the incoherent state tomography work~\cite{kuengfasttkomogr}: instead of applying a direct matrix-concentration bound that pays a logarithmic overhead, we control each fixed quadratic form by scalar Bernstein concentration and then use a constant-radius net of the unit sphere. 
This gives an operator-norm guarantee strong enough to imply the desired diamond-norm error without the logarithmic factor.
Together with the coherent upper and lower bounds, our result gives a strict quantum-memory separation for full process tomography between $\Theta(d_{\mathrm{in}}^2d_{\mathrm{out}}^2/\varepsilon^2)$ for coherent protocols and $\Theta(d_{\mathrm{in}}^3d_{\mathrm{out}}^3/\varepsilon^2)$ for incoherent protocols.

These results have direct implications for process tomography in practice.
The optimal incoherent rate is attained by a simple non-adaptive, ancilla-free protocol with randomized input states and measurements, while our lower bound allows full classical adaptivity and fresh entangled ancillas.
Consequently, when required to succeed on every channel, more sophisticated incoherent strategies, including adaptive designs, ancilla-assisted process tomography~\cite{dariano2001quantum,altepeter2003ancilla}, and direct characterization of quantum dynamics~\cite{mohseni2006direct,mohseni2008quantum}, can improve on randomized measurements by at most a constant factor.
Without quantum memory, further savings must therefore come from structural assumptions on the channel or from weaker figures of merit, such as predicting restricted sets of channel properties~\cite{huang2023learning,kunjummen2023shadow,levy2024classical}.
The advantage of quantum memory, in turn, is a factor $\Theta(\din\dout)$ independent of the target accuracy $\varepsilon$, so it is most significant for higher-dimensional processes, such as $n$-qubit operations with $\D=4^n$ or high-dimensional photonic channels~\cite{bouchard2019quantum}.
Realizing it requires storing the outputs of multiple channel uses coherently until a joint measurement, which is most natural when the unknown process can be interfaced with a quantum processor, as in quantum-enhanced experiments~\cite{aharonov2022quantum,Huang_2022}, or with quantum memories at the nodes of a quantum network~\cite{wehner2018quantum}.
Whether a limited or noisy quantum memory already yields an advantage remains open (see Section~\ref{s:discussion}).

\section{Preliminaries} \label{sec:preliminaries}

In this section, we introduce the basic notation and conventions. 
All Hilbert spaces considered in this work are finite-dimensional. 
The unknown channel maps $\calL(A)$ to $\calL(B)$, where $A\simeq\C^{\din}$ and $B\simeq\C^{\dout}$. 
We write $\D:=\din\dout$.
For a Hilbert space $H$, let $\calL(H)$ denote the linear operators on $H$ and let $\calD(H)$ denote the density operators. 
Schatten norms are denoted by $\normone{\cdot}$, $\normtwo{\cdot}$, and $\norminf{\cdot}$.

\subsection{Choi convention}

Fix a computational basis $\{\ket{i}\}_{i=1}^{\din}$ of $A$, and let $A'\simeq A$ be a reference copy with the same basis. 
Define
\begin{equation}
\ket{\Omega}_{A'A}:=\frac{1}{\sqrt{\din}}\sum_{i=1}^{\din}\ket{i}_{A'}\ket{i}_{A}.
\end{equation}
For a linear map $\Phi:\calL(A)\to\calL(B)$, we use the Choi characterization of completely positive trace-preserving maps in normalized form \cite{Choi1975,Watrous2018}
\begin{equation}
J_\Phi:=(\id_{A'}\otimes\Phi)(\proj{\Omega})=\frac{1}{\din}\sum_{i,j=1}^{\din}\ketbra{i}{j}_{A'}\otimes\Phi(\ketbra{i}{j}_{A}).
\end{equation}
After defining $J_\Phi$, we identify $A'$ with $A$ and regard $J_\Phi\in\calL(A\otimes B)$. 
All transposes are taken in the computational basis fixed above.
With this normalization, $\Phi$ is a quantum channel if and only if
\begin{equation}
J_\Phi\ge0,\qquad
\Tr_B J_\Phi=\frac{I_A}{\din}.
\end{equation}
We also use the diamond norm $\normdiam{\cdot}$ and the elementary implication
\begin{equation}
\normone{J_\Phi-J_\Psi}\le\normdiam{\Phi-\Psi}
\end{equation}
for any linear maps $\Phi,\Psi:\calL(A)\to\calL(B)$.

\subsection{Adaptive incoherent protocols}

An adaptive incoherent protocol queries the unknown channel once at a time, makes a measurement, and stores only classical information between different measurements. 
We consider a protocol with $T$ rounds of measurements.
At round $t\in\{1,\ldots,T\}$, after observing a classical history $h_{t-1}=(y_1,\ldots,y_{t-1})$, the learner chooses an ancilla $R_t$, an input state $\rho_{h_{t-1}}\in\calD(R_t\otimes A)$, and a discrete POVM $\{M_{h_{t-1},y}\}_{y\in\calY_{h_{t-1}}}$ on $R_t\otimes B$. 
Throughout the main proof, the outcome set $\calY_{h_{t-1}}$ is taken to be finite or countably infinite. 
The channel is applied once to the ancillary register $A$, the POVM is measured immediately, and the resulting outcome $Y_t$ is appended to the transcript. 
We only allow classical adaptivity and processing among different rounds of measurements.

The full transcript of such a protocol is given by $Z:=(Y_1,\ldots,Y_T)$, and the estimator $\widehat{\Phi}$ is a function of $Z$.
In particular, $\widehat{\Phi}$ need not be a channel. 
Randomized protocols are handled by conditioning on the private random seed. 
Consequently, it suffices to prove the lower bound for deterministic protocols and then average over the seed.

The restriction to discrete outcomes is only a notational simplification to make the transcript likelihoods point probabilities, and thus all sums over outcomes and transcripts are countable. 
The same results extend to arbitrary standard Borel outcome spaces by replacing point probabilities with Radon-Nikodym derivatives. Informally, a standard Borel space is an outcome space that can be relabelled as a finite or countable set, or as the real line, while preserving which events can be assigned probabilities.
We provide the corresponding measure-theoretic formulation in Appendix~\ref{a:borel}.

\subsection{Single-round testers}

A single incoherent use of a channel can be written in tester form, the one-slot case of the quantum strategy or comb formalism~\cite{GutoskiWatrous2007,ChiribellaDArianoPerinotti2009}. 
A discrete single-round tester is a family of positive operators
\begin{equation}
\{\mathsf{T}_y\}_{y\in\calY},\qquad
\mathsf{T}_y\in\calL(A\otimes B),
\end{equation}
such that, for some $\tau\in\calD(A)$,
\begin{equation}
\sum_{y\in\calY}\mathsf{T}_y=\tau^\top\otimes I_B .
\end{equation}
When the tested channel is $\Phi$, the outcome probabilities are
\begin{equation}
p_\Phi(y)=\din\,\Tr(\mathsf{T}_yJ_\Phi).
\end{equation}

In Section~\ref{sec:one-use-testers}, we prove the representation lemma showing that every single round of a potentially adaptive protocol with incoherent measurements induces such a tester.
Thus, along any fixed transcript $z=(y_1,\ldots,y_T)$ of an adaptive protocol, the operators used in the likelihood ratio are deterministic:
\begin{equation}
\mathsf{T}_t(z):=\mathsf{T}_{y_{<t},y_t},\qquad
y_{<t}:=(y_1,\ldots,y_{t-1}).
\end{equation}
Here $\mathsf{T}_{y_{<t},y_t}$ denotes the tester element for outcome $y_t$ at round $t$, given the history $y_{<t}$ up to round $t-1$.

\subsection{Bayesian notation}

For a channel family parametrized by $\{\Phi_\theta:\theta\in\Theta\}$ and a fixed deterministic protocol, let $P_\theta$ be the induced probability distribution over the transcript $Z$. For the discrete outcomes considered here, the \emph{transcript support} of $P_\theta$ is the set $\{z:P_\theta(z)>0\}$ of full transcripts that can occur with positive probability.
If $\mu$ is a prior on $\Theta$, the prior predictive distribution is
\begin{equation}
P_\mu(z):=\int_\Theta P_\theta(z)\,d\mu(\theta).
\end{equation}
For every transcript $z$ with $P_\mu(z)>0$, the posterior is
\begin{equation}
\nu_z(A):=\frac{\int_A P_\theta(z)\,d\mu(\theta)}{\int_\Theta P_\theta(z)\,d\mu(\theta)},\qquad
A\subseteq\Theta.
\end{equation}
When a reference parameter $\theta_0$ has common transcript support with all parameters under consideration, we also write
\begin{equation}
R_\theta^{\theta_0}(z):=\frac{P_\theta(z)}{P_{\theta_0}(z)}
\end{equation}
and express the posterior as
\begin{equation}
\nu_z(A)=\frac{\int_A R_\theta^{\theta_0}(z)\,d\mu(\theta)}{\int_\Theta R_\theta^{\theta_0}(z)\,d\mu(\theta)}.
\end{equation}

\section{Main results}
\label{subsec:results-overview}

Our main result is a lower bound for learning an arbitrary quantum channel in diamond norm for any adaptive incoherent process tomography protocol. 
The protocol may choose its input state and measurement at each round as an arbitrary function of the previous classical outcomes, and use an unbounded-size ancilla register within each round.
As the measurements are assumed incoherent, only classical adaptivity is allowed among different measurement rounds.
The proof strategy is summarized in Fig.~\ref{fig:proof_strategy}.

\begin{figure*}[t]
\centering
\includegraphics[width=0.97\textwidth]{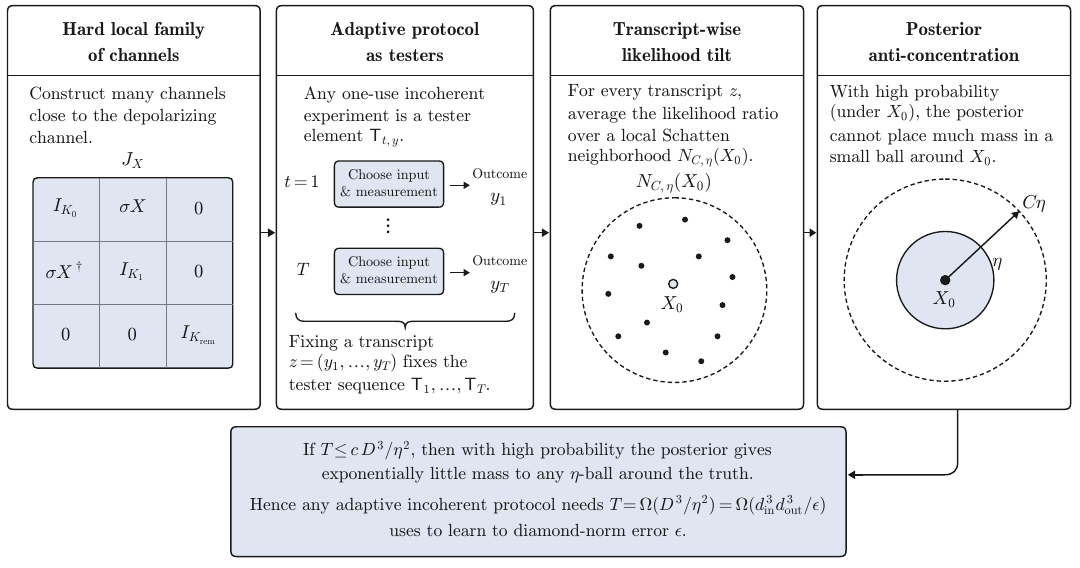}
\caption{\textbf{Overview of the proof strategy.} We first construct a local family of channels near the completely depolarizing channel.
Each channel is indexed by a matrix perturbation of the normalized Choi operator. 
We then represent every round of an adaptive incoherent protocol by a single-round tester. 
Along any transcript, the adaptivity protocol yields a deterministic sequence of tester elements.
The key technical ingredient is a transcript-wise likelihood tilt bound over local Schatten neighborhoods, which yields posterior anti-concentration: after fewer than order $\din^3\dout^3/\varepsilon^2$ channel uses, the posterior cannot concentrate in the Choi trace ball that would be forced by any diamond-norm accurate estimator. 
This gives the lower bound $T=\Omega(\din^3\dout^3/\varepsilon^2)$.
}
\label{fig:proof_strategy}
\end{figure*}  

\begin{theorem}[Adaptive incoherent process tomography lower bound]
\label{thm:main}
There exist universal constants $c>0$, $\varepsilon_\star>0$, and $\D_0\in\N$ such that the following holds. 
Let $A$ and $B$ be the system of $\C^{\din}$ and $\C^{\dout}$ dimensions with $\dout\ge 2$.
For simplicity, we denote $\D=\din\dout\ge \D_0$.
Suppose an adaptive incoherent protocol, allowing standard Borel outcome spaces in each measurement, uses $T$ queries to an unknown channel $\Phi:\calL(A)\to\calL(B)$ and outputs a linear-map-valued estimator $\widehat{\Phi}$. 
If, for every
channel $\Phi$,
\begin{equation}
\Prob_\Phi\left[\normdiam{\widehat{\Phi}-\Phi}\le\varepsilon\right]\ge \frac23
\end{equation}
for some $0<\varepsilon\le\varepsilon_\star$, then
\begin{equation}
T\ge c\,\frac{\D^3}{\varepsilon^2}=c\,\frac{\din^3\dout^3}{\varepsilon^2}.
\end{equation}
\end{theorem}

We emphasize that the lower bound applies to any incoherent strategy with classical adaptivity between rounds, arbitrary ancilla-assisted inputs within each round, and arbitrary outcome-dependent measurements. 
Theorem~\ref{thm:main} is complemented by a non-adaptive ancilla-free upper bound $O({\din^3\dout^3}/{\varepsilon^2})$, derived in Section~\ref{sec:log-free-upper-bound}. 
The upper bound is based on previously introduced algorithm, and our contribution is to remove the logarithmic factor in the previously best known bound~\cite{SurawyStepneyKahnKuengGuta2022,oufkir2023sample}. 
Therefore, for unknown channels $\Phi:\calL(\C^{\din})\to\calL(\C^{\dout})$, the optimal query complexity for diamond-norm incoherent process tomography is
\begin{equation}
\Theta\left(\frac{\din^3\dout^3}{\varepsilon^2}\right).
\end{equation}

Since incoherent protocols are allowed arbitrary classical adaptivity, any remaining gap between incoherent and coherent tomography must instead come from the ability to preserve and jointly process quantum information across channel uses.
Coherent process tomography protocols that may retain and jointly process quantum information across channel uses achieve the optimal query complexity
\begin{equation}
\Theta\left(\frac{\din^2\dout^2}{\varepsilon^2}\right)
\end{equation}
for the same task~\cite{mele2025optimal,chen2026quantum}. Hence, quantum memory improves the sample complexity by a factor of $R=\Theta(\din\dout)$ over any fully adaptive incoherent strategy. 
 
\section{Lower bound}
This section formalizes the proof sketch, while deferring most of the fully detailed proofs to the appendix.
We first reduce each single-round experiment to a tester acting on the Choi operator. 
We then construct the hard local family of channels and record the geometric properties of the prior. 
The key technical subroutine is the transcript-wise adaptive test tilt lemma, which controls likelihood ratios uniformly along a fixed adaptive transcript. 
The intuition is that too few observations leave many nearby channels consistent with the data.
The hard family has $\Theta(\D^2)$ real parameters, so enlarging a local neighborhood by a fixed factor gives a volume gain of $\exp(\Theta(\D^2))$.
The tester representation expresses the probability of each observed outcome in terms of the unknown Choi operator, and fixing a transcript fixes all adaptive experimental choices.
The likelihood tilt lemma then shows that, at a Choi trace-norm scale proportional to $\varepsilon$, $T$ queries can suppress the average likelihood over the enlarged neighborhood by at most a factor $\exp(O(T\varepsilon^2/\D))$.
The posterior comparison balances this likelihood loss against the geometric gain: when $T\varepsilon^2/\D$ is smaller than a sufficiently small constant times $\D^2$, the conditional distribution over channels remains too spread out for uniformly successful estimation.
Since diamond-norm accuracy also requires Choi trace-norm accuracy, this gives
\begin{equation}
    T=\Omega\left(\frac{\D^3}{\varepsilon^2}\right)
=\Omega\left(\frac{\din^3\dout^3}{\varepsilon^2}\right).
\end{equation}
After completing the lower-bound proof, Section~\ref{sec:log-free-upper-bound} states the matching non-adaptive incoherent upper bound.
\subsection{Single-round tester formalism}\label{sec:one-use-testers}

The first step is to separate the physical implementation of a single incoherent experiment from its statistical effect on the unknown channel. 
At a fixed history of an adaptive protocol, the learner has chosen an ancilla-assisted input state and a measurement on the output. 
The following lemma shows that, as a function of the channel, this entire single-round experiment is equivalently described by a collection of positive operators acting on the Choi space $A\otimes B$.

\begin{lemma}[Single-round tester representation]
\label{lem:tester}
Consider a query to a channel $\Phi:\calL(A)\to\calL(B)$ with an input state $\rho\in\calD(R\otimes A)$ and a discrete POVM $\{M_y\}_{y\in\calY}$ on $R\otimes B$. 
There exist positive operators and a state $\tau\in\calD(A)$ such that
\begin{equation}
\mathsf{T}_y\in\calL(A\otimes B),\quad
y\in\calY,\qquad
\sum_{y\in\calY}\mathsf{T}_y=\tau^\top\otimes I_B
\end{equation}
satisfies
\begin{equation}
\Prob_\Phi[Y=y]=\din\,\Tr(\mathsf{T}_yJ_\Phi).
\end{equation}
Here, the transpose is taken in the computational basis used to define $J_\Phi$.
\end{lemma}

\noindent The proof of Lemma~\ref{lem:tester} is given by a standard purification-and-vectorization argument and is provided in Appendix~\ref{app:one-use-tester-proof}. 
The key point is that the operators $\mathsf{T}_y$ depend only on the chosen experiment instead of the unknown channel. 
All dependence on $\Phi$ enters linearly through the normalized Choi operator $J_\Phi$.

Applying Lemma~\ref{lem:tester} at each history of an adaptive incoherent protocol gives the following transcript-wise representation. 
If $h=(y_1,\ldots,y_{t-1})$ is a possible history before round $t$, then the input state and POVM chosen by the protocol at $h$ induce tester elements $\{\mathsf{T}_{h,y}\}_{y\in\calY_h}$ satisfying
\begin{equation}
\sum_{y\in\calY_h}\mathsf{T}_{h,y}=\tau_h^\top\otimes I_B
\end{equation}
for some $\tau_h\in\calD(A)$, and
\begin{equation}
\Prob_\Phi[Y_t=y\mid h]=\din\,\Tr(\mathsf{T}_{h,y}J_\Phi).
\end{equation}
Consequently, once a full transcript $z=(y_1,\ldots,y_T)$ is fixed, the adaptive choices along that transcript become deterministic.
We write
\begin{equation}
\mathsf{T}_t(z):=\mathsf{T}_{y_{<t},y_t},\qquad
y_{<t}:=(y_1,\ldots,y_{t-1}).
\end{equation}
Thus, for any reference channel $\Phi_0$ and any transcript $z$ with $P_{\Phi_0}(z)>0$, the likelihood ratio factors as
\begin{equation}
\frac{P_\Phi(z)}{P_{\Phi_0}(z)}=\prod_{t=1}^T\frac{\Tr(\mathsf{T}_t(z)J_\Phi)}{\Tr(\mathsf{T}_t(z)J_{\Phi_0})}.
\label{eq:transcript-likelihood-ratio}
\end{equation}
The factors $\din$ cancel in the ratio.
Eq.~\eqref{eq:transcript-likelihood-ratio} is the formal mechanism by which we handle adaptivity. 
Although the protocol may choose each experiment as an arbitrary function of the past, conditioning on a transcript freezes those choices. 
The lower bound will therefore control adaptive protocols by proving a likelihood-ratio estimate that holds for every deterministic sequence of tester elements arising along a transcript.

\subsection{Construction of the hard channel family} \label{sec:hard-channel-family}

We now construct the local family of channels used in our main results. 
The family is centered at the completely depolarizing channel with the normalized Choi operator
\begin{equation}
\Phi_{\mathrm{dep}}(\rho)=\Tr(\rho)\frac{I_B}{\dout},\quad
J_{\mathrm{dep}}=\frac{I_{AB}}{\D}.
\end{equation}

The perturbations are chosen so that they preserve the trace-preserving constraint exactly, while creating a high-dimensional set of locally separated channels.
Fix a universal constant $0<\sigma\le 10^{-3}$, we decompose the output space as
\begin{equation}
\begin{aligned}
B=B_0\oplus B_1&\oplus B_{\mathrm{rem}},\\
\dim B_0=\dim B_1=&\;s:=\left\lfloor\frac{\dout}{2}\right\rfloor.
\end{aligned}
\end{equation}
Set
\begin{equation}
\begin{aligned}
&K_0:=A\otimes B_0,\qquad
K_1:=A\otimes B_1,\\
&r:=\dim K_0=\dim K_1=\din s.
\end{aligned}
\end{equation}
Since $\dout\ge2$, this auxiliary dimension satisfies
\begin{equation}
\frac{\D}{3}\le r\le\frac{\D}{2}.
\label{eq:r-dimension-comparison}
\end{equation}
After fixing orthonormal bases of $K_0$ and $K_1$, we identify both spaces with $\C^r$. 
For $X\in\C^{r\times r}$, interpreted as an operator
$K_1\to K_0$, define the Hermitian off-diagonal block operator
\begin{equation}
E(X):=\begin{pmatrix}
0 & X\\
X^\dagger & 0
\end{pmatrix}
\end{equation}
on $K_0\oplus K_1$, extended by zero on $A\otimes B_{\mathrm{rem}}$.
The defining feature of this embedding is that the perturbation is invisible to the partial trace over $B$. 
Thus it preserves the affine constraint $\Tr_BJ=I_A/\din$ for Choi operators of channels.

\begin{lemma}[Off-diagonal perturbations]
\label{lem:E}
For every $X\in\C^{r\times r}$, we have
\begin{align}
\Tr_BE(X)=0,\quad
\norminf{E(X)}=\norminf{X},
\end{align}
and thus $\normone{E(X)}=2\normone{X}$ and $\normtwo{E(X)}=\sqrt2\normtwo{X}$.
\end{lemma}

\noindent The proof is a direct block-matrix computation and is deferred to Appendix~\ref{app:hard-family-algebra}. 

We are now ready to define the local family. 
Let
\begin{equation}
\begin{aligned}
S:=\{X\in\C^{r\times r}:\norminf{X}\le4\},\\
G:=\{X\in\C^{r\times r}:\norminf{X}\le3\}.
\end{aligned}
\end{equation}
For $X\in S$ and $\sigma\geq 0$, define
\begin{equation}
J_X:=\frac{I_{AB}}{\D}+\frac{\sigma}{\D}E(X).
\label{eq:hard-family-choi}
\end{equation}

\begin{lemma}[Validity of the hard family]
\label{lem:valid-family}
For every $X\in S$, $J_X$ is the normalized Choi operator of a valid channel $\Phi_X:\calL(A)\to\calL(B)$.
Moreover, for every $X\in G$,
\begin{equation}
\frac{1-3\sigma}{\D}I_{AB}\preceq J_X\preceq\frac{1+3\sigma}{\D}I_{AB}.
\label{eq:JX-well-conditioned-G}
\end{equation}
For every $X\in S$,
\begin{equation}
J_X\succeq\frac{1-4\sigma}{\D}I_{AB}.
\label{eq:JX-well-conditioned-S}
\end{equation}
\end{lemma}
\noindent Therefore, the hard family is a local perturbation of the depolarizing channel inside the full-rank part of the channel set. 
The set $G$ will be the \enquote{regular} part of the prior: under the prior that we introduce in the subsequent Section~\ref{sec:prior-regularity-local-geometry}, we will sample an $X$ with an overwhelming probability in $G$.
The larger set $S$ leaves room around points in $G$ for the local perturbations used below.
This keeps the enlarged neighborhoods within the prior's support, where all channels are valid and uniformly full rank.

Finally, the parametrization converts Schatten trace distance in the matrix parameter $X$ exactly into Choi trace distance.

\begin{lemma}[Trace distance in the hard family]
\label{lem:choi-trace-distance}
For any $X_0,X_0+W\in S$,
\begin{equation}
\normone{J_{X_0+W}-J_{X_0}}=\frac{2\sigma}{\D}\normone{W}.
\label{eq:choi-trace-distance-hard-family}
\end{equation}
\end{lemma}

\noindent Consequently, a Choi trace ball of radius $\eta$ around $J_{X_0}$ corresponds to a Schatten-$1$ ball in the local parameter $W$ of radius
\begin{equation}
L_\eta:=\frac{\eta\D}{2\sigma}.
\label{eq:L-eta-definition}
\end{equation}
This exact conversion is the reason for using the off-diagonal embedding $E(X)$: it lets us phrase posterior concentration around the true channel as a local volume question in the matrix space $\C^{r\times r}$, where $r=\Theta(\D)$.

\subsection{Prior regularity and local geometry}\label{sec:prior-regularity-local-geometry}

We next put a prior on the hard family and present the local geometric facts used in the posterior argument. 
The prior is a truncated complex Ginibre measure on the matrix parameter $X$, which has two useful features: typical draws lie in the well-conditioned set $G$, and the prior density is sufficiently regular on its support $S$.

Quantitatively, let $\mu$ be the probability measure on $S=\{X\in\C^{r\times r}:\norminf{X}\le4\}$ with density
\begin{equation}
f_\mu(X)=\frac{1}{Z}\exp(-r\normtwo{X}^2)\mathbf 1_S(X)
\end{equation}
with respect to Lebesgue measure on $\C^{r\times r}\simeq\R^{2r^2}$, where
\begin{equation}
Z=\int_S \exp(-r\normtwo{X}^2)\,dX .
\end{equation}
Equivalently, $\mu$ is the complex Ginibre distribution with independent entries of density
\begin{equation}
z\mapsto \frac{r}{\pi}e^{-r|z|^2},\qquad
z\in\C,
\end{equation}
conditioned on $\norminf{X}\le4$.

\begin{lemma}[Prior regularity]
\label{lem:prior}
There exist universal constants $\cgin>0$, $\Cdens=4$, and an integer $r_{\mathrm{gin}}$ such that, for every $r\ge r_{\mathrm{gin}}$,
\begin{equation}
\mu(G)\ge1-e^{-\cgin r},\ G=\{X\in\C^{r\times r}:\norminf{X}\le3\}.
\end{equation}
Moreover, for every $r\ge1$ and every $X,X'\in S$,
\begin{equation}
\frac{f_\mu(X)}{f_\mu(X')}\le\exp(\Cdens \D^2).
\label{eq:prior-density-ratio}
\end{equation}
\end{lemma}

\noindent The proof, given in Appendix~\ref{app:prior-geometry}, uses only the standard operator-norm tail bound for complex Ginibre matrices. 
The first conclusion says that a draw from the prior is regular with overwhelming probability.
The second is a deliberately crude density-ratio bound costing only $\exp(O(\D^2))$, which is of the same order as the dimension of the local parameter space and is harmless in the final volume comparison.

We now define the local neighborhoods used to test posterior concentration.
For $X_0\in S$ and $\eta>0$, let
\begin{equation}
B_\eta(X_0):=\{X\in S:\normone{J_X-J_{X_0}}\le\eta\}
\end{equation}
be the set of matrices $X$ in the hard family whose Choi operators are within trace-norm distance $\eta$ of $J_{X_0}$. 
By Lemma~\ref{lem:choi-trace-distance},
\begin{equation}
B_\eta(X_0)=\{X_0+W\in S:\normone{W}\le L_\eta\},\ L_\eta:=\frac{\eta\D}{2\sigma}.
\label{eq:Beta-Leta}
\end{equation}
We also write
\begin{equation}
\widetilde B_\eta(X_0):=\{X_0+W:\normone{W}\le L_\eta\}
\end{equation}
for the unrestricted translate, so that
\begin{equation}
\vol(\widetilde B_\eta(X_0))=\vol\{W:\normone{W}\le L_\eta\}.
\end{equation}
The posterior anti-concentration argument will compare the small ball $B_\eta(X_0)$ with a larger, regularized neighborhood. 
For $C>1$, we define
\begin{equation}
N_{C,\eta}:=\left\{W\in\C^{r\times r}:\normone{W}\le C L_\eta,\;\norminf{W}\le\frac{C L_\eta}{4r}\right\}.
\label{eq:NCeta-definition}
\end{equation}
The condition $\normone{W}\le C L_\eta$ places $X_0+N_{C,\eta}$ inside the Choi trace ball $B_{C\eta}(X_0)$, while the additional operator-norm constraint ensures that every perturbation is pointwise small enough for the likelihood-ratio estimates below.

\begin{lemma}[Support preservation]
\label{lem:support}
If $\eta\le\tfrac{8\sigma}{3C}$, then for every $X_0\in G$ and every $W\in N_{C,\eta}$, we have
\begin{equation}
X_0+W\in S.
\end{equation}
\end{lemma}

\noindent Thus, for regular centers $X_0\in G$, the translated set $X_0+N_{C,\eta}$ remains inside the support of the prior. 
In particular, under the same condition on $\eta$,
\begin{equation}
X_0+N_{C,\eta}\subseteq B_{C\eta}(X_0).
    \label{eq:NCeta-contained-large-ball}
\end{equation}

The final geometry on the input is that the regularized enlargement $N_{C,\eta}$ has exponentially larger volume than the original trace ball once $C$ is chosen large enough, which results in the entropy source for the anti-concentration argument.

\begin{lemma}[Local Schatten volume estimate]
\label{lem:volume}
There exist universal constants $\cvol>0$ and $r_{\mathrm{vol}}\in\N$ such that, for every $r\ge r_{\mathrm{vol}}$,
\begin{equation}
\frac{\vol\{W:\normone{W}\le1,\ \norminf{W}\le1/(4r)\}}{\vol\{W:\normone{W}\le1\}}\ge e^{-\cvol r^2}.
\label{eq:unit-volume-estimate}
\end{equation}
Consequently, for every $C>1$ and every $\eta>0$,
\begin{equation}
\frac{\vol(N_{C,\eta})}{\vol\{W:\normone{W}\le L_\eta\}}\ge C^{2r^2}e^{-\cvol r^2}.
\label{eq:scaled-volume-estimate}
\end{equation}
\end{lemma}

\noindent The proof of Lemma~\ref{lem:volume}, deferred to Appendix~\ref{app:prior-geometry}, uses the standard volume-radius asymptotic approximation for complex Schatten balls~\cite{KabluchkoProchnoThale2020}. 
The exponent $2r^2$ is the dimension of $\C^{r\times r}$. 
Since $r=\Theta(\D)$, the volume gain from enlarging by a constant factor is of order $\exp(\Theta(\D^2))$, which will dominate the prior density-ratio loss in Eq.~\eqref{eq:prior-density-ratio} and the likelihood loss controlled in the next section.

\subsection{The transcript-wise adaptive tester tilt lemma}\label{sec:transcript-wise-tilt}

We now state the main technical estimate that handles adaptivity. 
A standard Bayesian lower bound might try to control the information gained in each round. 
However, such a single-step bound is not sufficient here by itself, because an adaptive protocol may choose later testers that are highly tuned to the posterior produced by earlier outcomes. 
We observe that, once the transcript is fixed, the adaptive protocol becomes a deterministic sequence of tester elements, as in Eq.~\eqref{eq:transcript-likelihood-ratio}. 
The lemma below controls the average likelihood ratio over a local Schatten neighborhood for every such frozen transcript.

Fixing a regular center $X_0\in G$ and a perturbation
$W\in N_{C,\eta}$, we write
\begin{equation}
\Delta_W:=J_{X_0+W}-J_{X_0}=\frac{\sigma}{\D}E(W).
\end{equation}
The assumptions in the next lemma ensure that by Lemma~\ref{lem:support}, $X_0+W$ remains in the support $S$ of the hard family, and that each single-step likelihood perturbation is uniformly small enough for quadratic logarithmic estimates to apply.

\begin{lemma}[Transcript-wise adaptive tester tilt]
\label{lem:tilt}
Fix $X_0\in G$ and $C>1$.
Assume $\eta>0$ satisfies
\begin{equation}
\eta\le\frac{8\sigma}{3C},\qquad
\frac{3C\eta}{8(1-3\sigma)}\le\frac{1}{10}.
\label{eq:tilt-eta-assumptions}
\end{equation}
Let $z=(y_1,\ldots,y_T)$ be a transcript with positive probability under $\Phi_{X_0}$. 
Along this transcript, let $\mathsf{T}_t(z):=\mathsf{T}_{y_{<t},y_t}$ with $t=1,\ldots,T$ be the tester elements induced by the adaptive protocol. 
For $W\in N_{C,\eta}$, define the transcript likelihood ratio
\begin{equation}
\Lambda_W(z):=\prod_{t=1}^T\frac{\Tr(\mathsf{T}_t(z)J_{X_0+W})}{\Tr(\mathsf{T}_t(z)J_{X_0})}.
\label{eq:Lambda-W-definition}
\end{equation}
Let $W$ be uniformly distributed on $N_{C,\eta}$. Then
\begin{equation}
\E_W[\Lambda_W(z)]\ge\exp\left(-\Ktilt\frac{C^2\eta^2T}{\D}\right),
\label{eq:tilt-bound}
\end{equation}
where we can take $\Ktilt=\tfrac{27}{2(1-3\sigma)^2}$.
\end{lemma}

\noindent We prove Lemma~\ref{lem:tilt} in Appendix~\ref{app:tilt-proof}. 
The proof has three ingredients. 
First, the operator-norm constraint in $N_{C,\eta}$ makes every single-step likelihood ratio close to one. 
Second, the off-diagonal embedding $E(W)$ reduces each linear fluctuation to a matrix inner product $\Tr(B_t^\dagger W)$ against the off-diagonal block of the tester. 
Third, the unitary invariance of $N_{C,\eta}$ gives an isotropic covariance bound for $W$, which yields a transcript-wise second-moment estimate. 
Jensen's inequality then converts this second-moment control into the lower bound in Eq.~\eqref{eq:tilt-bound}.

\noindent Lemma~\ref{lem:tilt} shows that, averaged over the symmetric neighborhood $N_{C,\eta}$, no fixed adaptive transcript can suppress the likelihood by more than $\exp(O(C^2\eta^2T/\D))$. 
The factor $1/\D$ per query to channel is the key quantitative gain. 
In the posterior argument, it will be compared against the volume growth
\begin{equation}
\frac{\vol(N_{C,\eta})}{\vol\{W:\normone{W}\le L_\eta\}}\ge C^{2r^2}e^{-\cvol r^2}.
\end{equation}
For a sufficiently large fixed $C$, the volume gain has an exponent of order $r^2=\Theta(\D^2)$, while the likelihood penalty in Eq.~\eqref{eq:tilt-bound} has an exponent of order $\eta^2T/\D$.
These exponents become comparable when $T$ is of order $\D^3/\eta^2$.
Below this scale, the volume gain dominates the likelihood loss.
The posterior comparison in Section~\ref{sec:posterior-anti-concentration} makes this precise and yields the lower-bound scaling
\begin{equation}
T=\Omega\left(\frac{\D^3}{\eta^2}\right).
\end{equation}

\subsection{Posterior anti-concentration}\label{sec:posterior-anti-concentration}

We now convert the transcript-wise tilt estimate into the main Bayesian statement used in the main result. 
For $X\in S$, let $P_X$ denote the law of the full transcript $Z$ when the unknown channel is $\Phi_X$. 
The first observation is that all channels in the hard family induce the same transcript support.

\begin{lemma}[Common transcript support]
\label{lem:common-support}
For every $X,X'\in S$ and every full transcript $z$,
\begin{equation}
P_X(z)>0\qquad\Longleftrightarrow\qquad P_{X'}(z)>0 .
\end{equation}
Consequently, for any fixed $X_0\in S$, the transcript laws $P_X$ and $P_{X_0}$ are mutually absolutely continuous on the discrete transcript space for every $X\in S$.
\end{lemma}

\noindent The proof is given in Appendix~\ref{app:posterior-anti-concentration-proof}. 
It uses only the uniform full-rank lower bound in Eq.~\eqref{eq:JX-well-conditioned-S}: at any history, whether a single-step outcome has positive probability is determined only by whether the corresponding tester element is nonzero, and not by the choice of $X\in S$.

Fixing $X_0\in G$, for a transcript $z$ in the common support, we define the likelihood ratio as
\begin{equation}
R_X(z):=\frac{P_X(z)}{P_{X_0}(z)} .
\end{equation}
By Lemma~\ref{lem:common-support}, $R_X(z)$ is positive and finite for all $X\in S$ on this support. 
Moreover,
\begin{equation}
\E_{Z\sim P_{X_0}}R_X(Z)=1 .
\end{equation}
Let $X\sim\mu$ and, conditional on $X$, let $Z\sim P_X$. 
For every transcript $z$ with positive prior predictive probability, we have
\begin{equation}
P_\mu(z):=\int_S P_X(z)f_\mu(X)\,dX,
\end{equation}
the posterior is
\begin{equation}
\nu_z(A):=\frac{\int_{A\cap S}P_X(z)f_\mu(X)\,dX}{\int_S P_X(z)f_\mu(X)\,dX},\quad
A\subseteq S .
\end{equation}
Equivalently, for a fixed $X_0\in G$ and a transcript $z$ in the common support,
\begin{equation}
\nu_z(A)=\frac{\int_{A\cap S}R_X(z)f_\mu(X)\,dX}{\int_S R_X(z)f_\mu(X)\,dX}.
\end{equation}

The following lemma is the posterior anti-concentration statement, which indicates that, below the sample size $\D^3/\eta^2$, the posterior cannot assign substantial mass to a Choi trace ball of radius $\eta$ around the true channel, with high probability over transcripts generated by that true channel.

\begin{lemma}[Adaptive posterior anti-concentration]
\label{lem:anti}
There exist universal constants $C>1$, $\eta_\star>0$, $\cac>0$, $\kappa>0$, and an integer $r_\star\in\N$ such that the following holds for every
$r\ge r_\star$. 
Let
\begin{equation}
0<\eta\le\eta_\star,\qquad
T\le \cac\frac{\D^3}{\eta^2}.
\end{equation}
Then, for every $X_0\in G$,
\begin{equation}
P_{X_0}\left[\nu_Z(B_\eta(X_0))\le e^{-\kappa\D^2}\right]\ge 1-e^{-\D^2}.
\label{eq:posterior-anti-concentration}
\end{equation}
\end{lemma}

\noindent The proof of Lemma~\ref{lem:anti} is deferred to Appendix~\ref{app:posterior-anti-concentration-proof}. 
The argument compares two posterior masses around the same center $X_0$. 
First, a Markov inequality shows that the unweighted likelihood integral over the small ball $B_\eta(X_0)$ is rarely much larger than its volume. 
Second, the transcript-wise tilt lemma lower bounds the likelihood integral over the larger regularized neighborhood $X_0+N_{C,\eta}\subseteq B_{C\eta}(X_0)$.
The volume gain from $N_{C,\eta}$ is exponential in $r^2=\Theta(\D^2)$, while the likelihood loss is at most of scaling
\begin{equation}
\exp\left(O\left(\frac{\eta^2T}{\D}\right)\right).
\end{equation}
Thus, when $T\lesssim \D^3/\eta^2$, the larger neighborhood has much more posterior mass than the smaller ball. 
Since the larger posterior mass is at most one, the smaller posterior mass must be exponentially small.

This anti-concentration lemma is the final building block of the proof for Theorem~\ref{thm:main}. 
In the next section, we show that any uniformly accurate estimator would force the posterior to place non-negligible mass in one such small ball, contradicting Eq.~\eqref{eq:posterior-anti-concentration}.

\subsection{Proof of the main theorem}\label{sec:proof-main-theorem}

We now assemble the preceding ingredients to prove Theorem~\ref{thm:main}.

\begin{proof}[Proof of Theorem~\ref{thm:main}]
It suffices to prove the claim for deterministic protocols, as any randomized protocol is a mixture of deterministic protocols indexed by a private random seed, independent of the unknown channel. 
Since the argument below applies uniformly to every fixed seed, averaging over the seed gives the same conclusion for randomized protocols.

Let $C,\eta_\star,\cac,\kappa,r_\star$ be the constants from Lemma~\ref{lem:anti}, and let $\cgin$ be the constant from Lemma~\ref{lem:prior}. 
Set
\begin{equation}
\varepsilon_\star:=\frac{\eta_\star}{2},\qquad
c:=\frac{\cac}{4}.
\end{equation}
Choose $\D_0$ large enough so that, whenever $\D\ge \D_0$, one has $r\ge r_\star$ and
\begin{equation}
e^{-\cgin r}+e^{-\D^2}+e^{-\kappa\D^2}\le\frac13.
\label{eq:error-terms-main-proof}
\end{equation}

Assume, for contradiction, that a deterministic adaptive incoherent protocol satisfies the uniform guarantee in Theorem~\ref{thm:main} with
\begin{equation}
T\le c\frac{\D^3}{\varepsilon^2},\qquad
0<\varepsilon\le\varepsilon_\star .
\end{equation}
Draw $X\sim\mu$, set the unknown channel to $\Phi_X$, and let $Z\sim P_X$ be the transcript. 
Apply Lemma~\ref{lem:anti} with $\eta:=2\varepsilon$.
Since $\eta\le\eta_\star$ and
\begin{equation}
T\le\frac{\cac}{4}\frac{\D^3}{\varepsilon^2}=\cac\frac{\D^3}{(2\varepsilon)^2}=\cac\frac{\D^3}{\eta^2},
\end{equation}
the lemma applies. 
Therefore, for every $X_0\in G$,
\begin{equation}
P_{X_0}\left[\nu_Z(B_{2\varepsilon}(X_0))>e^{-\kappa\D^2}\right]\le e^{-\D^2}. \label{eq:cond_bound}
\end{equation}
Conditional on $X=X_0$, the transcript has distribution $P_{X_0}$.
Integrating Eq.~\eqref{eq:cond_bound} over $X_0\in G$ gives an upper bound of $\mu(G)e^{-\D^2}\le e^{-\D^2}$ and therefore
\begin{align}
\begin{split}
&\Prob\left[X\in G\ \text{and}\ \nu_Z(B_{2\varepsilon}(X))>e^{-\kappa\D^2}\right]\le e^{-\D^2},\\
&\Prob\left[X\notin G\right]\le e^{-\cgin r},
\label{eq:posterior-bad-integrated}
\end{split}
\end{align}
where the second inequality follows from $\Prob[X\notin G]=1-\mu(G)\le e^{-\cgin r}$ by Lemma~\ref{lem:prior}.

We now show that posterior anti-concentration rules out Bayes success in Choi trace norm. 
For a transcript $z$, define
\begin{align}
\begin{split}
A_z :=\Bigl\{X\in G:&\ \nu_z(B_{2\varepsilon}(X))\le e^{-\kappa\D^2}, \\
&\normone{J_{\widehat\Phi(z)}-J_X}\le\varepsilon\Bigr\}.
\end{split}\label{eq:AZ}
\end{align}
If $A_z$ is empty, then $\nu_z(A_z)=0$. 
Otherwise, one may choose $X_z\in A_z$. 
For every $X\in A_z$, the triangle inequality gives
\begin{equation}
\normone{J_X-J_{X_z}}\le\normone{J_X-J_{\widehat\Phi(z)}}+\normone{J_{\widehat\Phi(z)}-J_{X_z}}\le 2\varepsilon.
\end{equation}
Thus $A_z\subseteq B_{2\varepsilon}(X_z)$. Since $X_z\in A_z$,
\begin{equation}
\nu_z(B_{2\varepsilon}(X_z))\le e^{-\kappa\D^2}\quad\Rightarrow\quad
\nu_z(A_z)\le e^{-\kappa\D^2}.
\end{equation}
Averaging over the prior predictive law of $Z$ yields
\begin{equation}
\Prob[X\in A_Z]\le e^{-\kappa\D^2}.
\label{eq:choi-success-good-anti}
\end{equation}

Let
\begin{align}
\begin{split}
\mathcal E_{\mathrm{Choi}}&:=\left\{\normone{J_{\widehat\Phi(Z)}-J_X}\le\varepsilon\right\},\\
\mathcal E_{\mathrm{post}}&:=\left\{X\in G,\ \nu_Z(B_{2\varepsilon}(X))>e^{-\kappa\D^2}\right\}.
\end{split}
\end{align}
If $\mathcal E_{\mathrm{Choi}}$ occurs while $X\in G$ and $\mathcal E_{\mathrm{post}}$ does not occur, then both defining inequalities of $A_Z$ hold, so $X\in A_Z$.
Consequently,
\begin{equation}
\mathcal E_{\mathrm{Choi}}\subseteq\{X\notin G\}\cup\mathcal E_{\mathrm{post}}\cup\{X\in A_Z\}.
\end{equation}
Therefore, we have
\begin{equation}
\begin{aligned}
\Prob[\mathcal E_{\mathrm{Choi}}]&\le\Prob[X\notin G]+\Prob[\mathcal E_{\mathrm{post}}]+\Prob[X\in A_Z]  \\
&\le e^{-\cgin r}+e^{-\D^2}+e^{-\kappa\D^2}  \\
&\le \frac13,
\end{aligned}
\label{eq:bayes-choi-upper}
\end{equation}
where the last step follows from Eq.~\eqref{eq:error-terms-main-proof}.

On the other hand, applying the uniform success guarantee assumed in Theorem~\ref{thm:main} to $\Phi_X$ gives, for every $X\in S$,
\begin{equation}
\Prob_{\Phi_X}\left[\normdiam{\widehat\Phi-\Phi_X}\le\varepsilon\right]\ge\frac23.
\end{equation}
Averaging over $X\sim\mu$,
\begin{equation}
\Prob\left[\normdiam{\widehat\Phi(Z)-\Phi_X}\le\varepsilon\right]\ge\frac23.
\end{equation}
Since
\begin{equation}
\normone{J_{\widehat\Phi(Z)}-J_X}\le\normdiam{\widehat\Phi(Z)-\Phi_X},
\end{equation}
this implies
\begin{equation}
\Prob\left[\normone{J_{\widehat\Phi(Z)}-J_X}\le\varepsilon\right]\ge\frac23,
\end{equation}
contradicting Eq.~\eqref{eq:bayes-choi-upper}.

Therefore, no adaptive incoherent protocol satisfying the stated worst-case success guarantee can use $T\le c\D^3/\varepsilon^2$ channel queries. 
Equivalently, after decreasing $c$ by a universal factor if necessary, every such protocol satisfies
\begin{equation}
T\ge c\,\frac{\D^3}{\varepsilon^2}=c\,\frac{\din^3\dout^3}{\varepsilon^2}.
\end{equation}
This proves the main result.
\end{proof}

\section{A matching non-adaptive upper bound}
\label{sec:log-free-upper-bound}

In this section, we show that the scaling of Theorem~\ref{thm:main} is tight \emph{up to a constant factor} by deriving a matching upper bound. 
The construction builds on the covariant projected-least-squares approach to incoherent process tomography~\cite{SurawyStepneyKahnKuengGuta2022,oufkir2023sample}, but replaces the direct matrix-concentration step by scalar concentration for fixed quadratic forms followed by a constant-radius covering argument. 
This removes the logarithmic factor in the previously best known upper bound.

\begin{theorem}[Incoherent process tomography upper bound]
\label{thm:log-free-upper}
There exists a universal constant $C>0$ such that the following holds. 
Let $\Phi:\calL(A)\to\calL(B)$ be an arbitrary quantum channel, let $0<\varepsilon\le1$, and let $0<\delta<1$. 
There is a non-adaptive, ancilla-free incoherent protocol that uses at most
\begin{equation}
\left\lceil C\,\frac{\D^3+\D^2\log(1/\delta)}{\varepsilon^2}\right\rceil
\label{eq:upper-bound-main-complexity}
\end{equation}
queries to $\Phi$ and outputs a quantum channel $\widetilde\Phi$ satisfying
\begin{equation}
\Prob_{\Phi}\left[\normdiam{\widetilde\Phi-\Phi}\le\varepsilon\right]\ge1-\delta.
\label{eq:upper-bound-main-guarantee}
\end{equation}
In particular, for constant success probability, the protocol uses
\begin{equation}
O\left(\frac{\D^3}{\varepsilon^2}\right)
=O\left(\frac{\din^3\dout^3}{\varepsilon^2}\right)
\label{eq:upper-bound-main-constant-confidence}
\end{equation}
channel queries.
\end{theorem}

\noindent The protocol and the proof of Theorem~\ref{thm:log-free-upper} are given in Appendix~\ref{app:log-free-upper-bound}. 
The estimator is an unbiased single-query estimator of the normalized Choi operator. 
The key observation is that every fixed quadratic form has variance bounded by a universal constant. 
Scalar Bernstein concentration therefore gives a tail of order $\exp(-\Omega(T\varepsilon^2/\D^2))$ in each fixed direction. 
A constant-radius net of the unit sphere has cardinality $\exp(O(\D))$, and the resulting union bound requires only $T=O(\D^3/\varepsilon^2)$ queries.

\begin{corollary}[Optimal incoherent sample complexity]
\label{cor:optimal-incoherent-complexity}
Let $\dout\ge2$, $\D\ge\D_0$, and $0<\varepsilon\le\min\{\varepsilon_\star,1\}$, where $\D_0$ and $\varepsilon_\star$ are the constants in Theorem~\ref{thm:main}. 
Then the query complexity of learning an arbitrary channel $\Phi:\calL(A)\to\calL(B)$ to diamond-norm error $\varepsilon$ with a constant success probability by adaptive incoherent protocols is
\begin{equation}
\Theta\left(\frac{\din^3\dout^3}{\varepsilon^2}\right).
\label{eq:optimal-incoherent-complexity}
\end{equation}
The upper bound is achieved by a non-adaptive ancilla-free protocol, whereas the lower bound continues to hold for fully adaptive protocols with arbitrary ancilla-assisted inputs within each round.
\end{corollary}

\section{Discussion} \label{s:discussion}

We have established that quantum process tomography exhibits a strict quantum memory advantage. 
In the model considered here, the learner may use arbitrary ancilla-assisted inputs within each channel query, perform arbitrary measurements after each use, carry out classical computation, and choose each future experiment as an arbitrary function of the entire previous transcript. 
The only forbidden resource is the ability to preserve quantum information across different uses of the unknown channel. 
Our main result shows that every such adaptive incoherent protocol requires $\Omega\left(\din^3\dout^3/\varepsilon^2\right)$ queries. 
Together with the non-adaptive ancilla-free upper bound in Theorem~\ref{thm:log-free-upper}, this gives the optimal incoherent query complexity
\begin{equation}
\Theta\left(\frac{\din^3\dout^3}{\varepsilon^2}\right).
\end{equation}
For $\din=1$, this recovers the optimal sample-complexity bounds for single-copy state tomography~\cite{chen2023does,kuengfasttkomogr}. Since coherent protocols achieve the optimal query complexity
$\Theta(d_{\rm in}^2d_{\rm out}^2/\varepsilon^2)$~\cite{mele2025optimal,chen2026quantum}, our results isolate quantum memory as the resource responsible for the separation.

We expect the techniques developed here for proving lower bounds against adaptive incoherent protocols in general channel tomography to be useful more broadly. Many quantum learning settings share the same underlying structure: each round consists of a quantum experiment, but only a classical outcome is retained before the next experiment is chosen. Natural directions include extending our methods to structured families of channels, unitaries, and dynamical processes, such as fermionic linear-optical and bosonic Gaussian processes~\cite{christensen2026learningfermioniclinearoptics,fanizza2026efficientlearningbosonicgaussian}, Pauli channels~\cite{fawzi2025lower}, and adaptive single-copy POVMs tomography~\cite{zambrano2025fast}.
Apart from standard volume and prior estimates, our argument uses only three features of the problem: in each round, the outcome probabilities are linear in an unknown positive operator subject to affine constraints, conditioning on a transcript freezes all adaptive choices, and the hard family lies in the interior of the constraint set, with perturbations that preserve the affine constraints exactly.
We therefore expect it to extend to the characterization of quantum detectors and instruments~\cite{lundeen2009tomography} and to multi-time processes described by quantum combs or process tensors~\cite{ChiribellaDArianoPerinotti2009,pollock2018non}, with multi-slot testers taking the place of single-round testers.

Several other directions remain open. First, it would be interesting to characterize the intermediate regimes between fully incoherent and fully coherent process tomography, in analogy with recent memory tradeoffs for state and Pauli-channel learning~\cite{chen2022exponential,chen2022quantum,chen2024tradeoff}. In particular, how much quantum memory is required to obtain a nontrivial improvement over the optimal incoherent rate?
A second direction is to determine the full query complexity of incoherent channel tomography as a function of the Kraus rank. Notably, adaptive lower bounds remain open even for bounded-rank state tomography~\cite{chen2023does}. Establishing corresponding lower bounds for low-Kraus-rank channels therefore appears to require overcoming difficulties that are already present in the simpler state-tomography setting.
This question is particularly relevant for the characterization of quantum gates, since the hard instances in our lower bound are highly mixed channels close to the completely depolarizing channel, whereas high-fidelity gates are close to unitary.
Finally, another important direction is to study process tomography with limited or noisy quantum memory, i.e., when only a limited number of coherent qubits can be stored across channel uses, and determine when this already suffices to obtain an advantage over incoherent protocols~\cite{chen2025efficient,arunachalam2026optimal}.

\subsection*{Data availability}
No data was generated in this work.

\subsection*{Author contributions}
All authors have contributed to deriving the results and writing the manuscript

\subsection*{Competing interests}
The authors declare no competing interests.

\subsection*{Acknowledgments}
    We thank Richard Allen, Senrui Chen, Sitan Chen, Angus Lowe, and Angelos Pelecanos for helpful discussions. 
    C.B.-P. and A.A.M. acknowledge the BMFTR (MUNIQC-Atoms, HYBRID, QuSol, Hybrid++), the DFG (CRC 183 and SPP 2514), the Quantum Flagship (Millenion, PasQuanS2), the Munich Quantum Valley, Berlin Quantum and the European Research Council (DebuQC) for financial support. A.A.M. further acknowledges support from a 2025 Google PhD Fellowship.
    W.G. is supported by NSF Grant CCF-2430375 and the Von Neumann Award from Harvard Computer Science.
    We acknowledge ChatGPT for assistance in discussing and refining proof ideas, as well as improving the presentation. The authors are solely responsible for the proofs and the results.
    
\bibliography{references}

\vspace*{0.5cm}

\onecolumngrid
\appendix

\newpage

\setcounter{figure}{0}
\setcounter{theorem}{0}

\begin{center}
\large{Supplementary Material for ``Quantum memory advantage for quantum process tomography''
}
\end{center}
\counterwithin{figure}{section}

\section{Proof of the single-round tester representation} \label{app:one-use-tester-proof}

\begin{proof}[Proof of Lemma~\ref{lem:tester}]
We first show that one can always reduce a general protocol to a protocol with only pure inputs of the same complexity. 
To see this, we consider an input $\rho\in\calD(R\otimes A)$.
Choose a purification of $\rho$ on an enlarged ancilla $R'R$. 
Measuring $I_{R'}\otimes M_y$ after the channel gives exactly the same outcome distribution as the original experiment. 
Thus it suffices to consider a pure input state $\ket{\psi}\in R\otimes A$, after absorbing the purifying register into $R$.

We first note that very vector $\ket{\psi}\in R\otimes A$ can be written uniquely as
\begin{equation}
\ket{\psi}_{RA}=(V\otimes I_A)\ket{\widetilde\Omega}_{A'A},\qquad \ket{\widetilde\Omega}_{A'A}:=\sum_{i=1}^{\din}\ket{i}_{A'}\ket{i}_{A}
\end{equation}
for a linear map $V:A'\to R$. 
Let $\tau:=\Tr_R\proj{\psi}\in\calD(A)$, we writing $\ket{\psi}=\sum_{i=1}^{\din}\ket{v_i}_R\ket{i}_A$ with $V\ket{i}=\ket{v_i}$. 
We have
\begin{equation}
\tau=\sum_{i,j=1}^{\din}\langle v_j|v_i\rangle\,\ketbra{i}{j},\qquad
V^\dagger V=\tau^\top.
\end{equation}
Let $\Gamma_\Phi:=\din J_\Phi$ be the unnormalized Choi operator. 
By the Choi identity,
\begin{equation}
(\id_R\otimes\Phi)(\proj{\psi})=(V\otimes I_B)\Gamma_\Phi(V^\dagger\otimes I_B).
\end{equation}
Define
\begin{equation}
\mathsf{T}_y:=(V^\dagger\otimes I_B)M_y(V\otimes I_B)\ge0,
\end{equation}
we have
\begin{equation}
\sum_{y\in\calY}\mathsf{T}_y=(V^\dagger\otimes I_B)\left(\sum_{y\in\calY}M_y\right)(V\otimes I_B)=(V^\dagger V)\otimes I_B=\tau^\top\otimes I_B.
\end{equation}
If $\calY$ is countably infinite, the equality is understood as the bound on the norm of increasing finite partial sums given that $A\otimes B$ is finite-dimensional.

Finally, by cyclicity of trace, we have
\begin{align}
\Prob_\Phi[Y=y]=\Tr\left[M_y(V\otimes I_B)\Gamma_\Phi(V^\dagger\otimes I_B)\right] \nonumber=\Tr(\mathsf{T}_y\Gamma_\Phi)=\din\,\Tr(\mathsf{T}_yJ_\Phi).
\end{align}
This proves the claim.
\end{proof}

\section{Algebra of the hard channel family}\label{app:hard-family-algebra}

We collect the elementary properties for the hard family introduced in Section~\ref{sec:hard-channel-family}.
We first observe that, since $s=\lfloor \dout/2\rfloor$,
\begin{equation}
r=\din s\le \frac{\din\dout}{2}=\frac{\D}{2}.
\end{equation}
For the lower bound, every integer $m\ge2$ satisfies $\lfloor\tfrac{m}{2}\rfloor\ge\tfrac{m}{3}$.
Therefore,
\begin{equation}
r=\din\left\lfloor\frac{\dout}{2}\right\rfloor\ge\frac{\din\dout}{3}=\frac{\D}{3}.
\end{equation}
Given that $\tfrac{\D}{3}\le r\le\frac{\D}{2}$, we now prove Lemma~\ref{lem:E}, Lemma~\ref{lem:valid-family}, and Lemma~\ref{lem:choi-trace-distance}.
\begin{proof}[Proof of Lemma~\ref{lem:E}]
For the partial trace, choose orthonormal bases of $B_0$, $B_1$, and $B_{\mathrm{rem}}$, and combine them into an orthonormal basis of $B$. 
For $a,a'\in A$, and $\{b\}$ as a set of basis for $B$, we have
\begin{equation}
\langle a|\Tr_BE(X)|a'\rangle=\sum_b\langle a,b|E(X)|a',b\rangle.
\end{equation}
The operator $E(X)$ has matrix elements only between $A\otimes B_0$ and $A\otimes B_1$. 
These subspaces are orthogonal in the $B$ register, so every summand with the same basis vector $b$ on the left and right vanishes. 
Hence, we conclude that
\begin{equation}
\Tr_BE(X)=0.
\end{equation}
For the norm identities, relative to $K_0\oplus K_1$,
\begin{equation}
E(X)^\dagger E(X)=\begin{pmatrix}
XX^\dagger & 0\\
0 & X^\dagger X
\end{pmatrix}.
\end{equation}
Thus the singular values of $E(X)$ are the singular values of $X$, with each repeated twice. 
The identities
\begin{equation}
\norminf{E(X)}=\norminf{X},\qquad
\normone{E(X)}=2\normone{X},\qquad
\normtwo{E(X)}=\sqrt2\normtwo{X}
\end{equation}
follow immediately from matrix norm inequalities.
\end{proof}

\begin{proof}[Proof of Lemma~\ref{lem:valid-family}]
Let $X\in S$. Since $\norminf{X}\le4$, Lemma~\ref{lem:E} gives $\norminf{E(X)}\le4$.
Therefore, we have
\begin{equation}
J_X=\frac{I_{AB}}{\D}+\frac{\sigma}{\D}E(X)\ge\frac{1-4\sigma}{\D}I_{AB}.
\end{equation}
Because $\sigma\le10^{-3}$, this operator is positive definite.
Next, using Lemma~\ref{lem:E}, we have
\begin{equation}
\Tr_BJ_X=\Tr_B\left(\frac{I_A\otimes I_B}{\D}\right)+\frac{\sigma}{\D}\Tr_BE(X)=\frac{\dout}{\din\dout}I_A=\frac{I_A}{\din}.
\end{equation}
Thus $J_X$ is the normalized Choi operator of a channel by the Choi characterization.
If $X\in G$, then $\norminf{X}\le3$, and hence $-3I_{AB}\le E(X)\le 3I_{AB}$.
Substituting this into the definition of $J_X$ gives
\begin{equation}
\frac{1-3\sigma}{\D}I_{AB}\preceq J_X\preceq\frac{1+3\sigma}{\D}I_{AB}.
\end{equation}
\end{proof}

\begin{proof}[Proof of Lemma~\ref{lem:choi-trace-distance}]
We have
\begin{equation}
J_{X_0+W}-J_{X_0}=\frac{\sigma}{\D}E(W).
\end{equation}
By Lemma~\ref{lem:E}, $\normone{E(W)}=2\normone{W}$, and thus 
\begin{equation}
\normone{J_{X_0+W}-J_{X_0}}=\frac{2\sigma}{\D}\normone{W}.
\end{equation}
\end{proof}

\section{Proofs for prior regularity and local geometry}\label{app:prior-geometry}

\begin{proof}[Proof of Lemma~\ref{lem:prior}]
Let $X_{\mathrm{Gin}}$ be an unconditioned complex Ginibre matrix with independent entries of density
\begin{equation}
z\mapsto \frac{r}{\pi}e^{-r|z|^2}.
\end{equation}
Equivalently, $\E |(X_{\mathrm{Gin}})_{ij}|^2=1/r$. 
The standard operator-norm tail bound for complex Ginibre matrices gives universal constants $c>0$ and $r_0\in\N$ such that, for all $r\ge r_0$,
\begin{equation}
\Prob\left[\norminf{X_{\mathrm{Gin}}}>2+\frac{t}{\sqrt r}\right]\le 2e^{-ct^2}\qquad\text{for all }t\ge0 .
\end{equation}
Taking $t=\sqrt r$ yields
\begin{equation}
\Prob[\norminf{X_{\mathrm{Gin}}}>3]\le 2e^{-cr}.
\end{equation}
After decreasing the constant and increasing $r_{\mathrm{gin}}$ if necessary, we may write this bound as $e^{-\cgin r}$.
Since $\mu$ is the Ginibre law conditioned on $\norminf{X}\le4$,
\begin{equation}
\mu(G)=\frac{\Prob[\norminf{X_{\mathrm{Gin}}}\le3]}{\Prob[\norminf{X_{\mathrm{Gin}}}\le4]}\ge\Prob[\norminf{X_{\mathrm{Gin}}}\le3]\ge1-e^{-\cgin r}.
\end{equation}
For the density-ratio bound, if $X,X'\in S$, then
\begin{equation}
\frac{f_\mu(X)}{f_\mu(X')}=\exp\left(-r\normtwo{X}^2+r\normtwo{X'}^2\right)\le\exp\left(r\normtwo{X'}^2\right).
\end{equation}
Since $\norminf{X'}\le4$ and $X'\in\C^{r\times r}$, we have $\normtwo{X'}^2\le r\norminf{X'}^2\le 16r$, and thus
\begin{equation}
\frac{f_\mu(X)}{f_\mu(X')}\le e^{16r^2}\le e^{4\D^2}
\end{equation}
using $r\le \D/2$, we have $16r^2\le4\D^2$.
This proves the claim with $\Cdens=4$.
\end{proof}

\begin{proof}[Proof of Lemma~\ref{lem:support}]
Let $X_0\in G$ and $W\in N_{C,\eta}$. 
By definition of $N_{C,\eta}$ and $r\ge \D/3$, we have
\begin{equation}
\norminf{W}\le\frac{C L_\eta}{4r}=\frac{C\eta\D}{8\sigma r}\le\frac{3C\eta}{8\sigma}.
\end{equation}
If $\eta\le 8\sigma/(3C)$, then $\norminf{W}\le1$. Since $X_0\in G$, $\norminf{X_0}\le3$, and therefore
\begin{equation}
\norminf{X_0+W}\le\norminf{X_0}+\norminf{W}\le4.
\end{equation}
Thus $X_0+W\in S$.
\end{proof}

\begin{proof}[Proof of Lemma~\ref{lem:volume}]
Let $B_1^r, B_\infty^r\subseteq \C^{r\times r}\simeq\R^{2r^2}$ with
\begin{equation}
B_1^r:=\{W:\normone{W}\le1\},\qquad
B_\infty^r:=\{W:\norminf{W}\le1\}.
\end{equation}
The complex Schatten-ball volume-radius asymptotics imply that there are universal constants $0<a<1$ and $r_{\mathrm{vol}}\in\N$ such that, for every $r\ge r_{\mathrm{vol}}$,
\begin{equation}
\left(\frac{\vol(B_\infty^r)}{\vol(B_1^r)}\right)^{1/(2r^2)}\ge ar.
\end{equation}
As $\norminf{W}\le1/(4r)$ implies $\normone{W}\le r\norminf{W}\le\frac14\le1$, we have
\begin{equation}
\frac{1}{4r}B_\infty^r\subseteq\{W:\normone{W}\le1,\ \norminf{W}\le1/(4r)\}
\end{equation}
Therefore, we compute
\begin{equation}
\frac{\vol\{W:\normone{W}\le1,\ \norminf{W}\le1/(4r)\}}{\vol(B_1^r)}\ge(4r)^{-2r^2}\frac{\vol(B_\infty^r)}{\vol(B_1^r)}\ge(4r)^{-2r^2}(ar)^{2r^2}=\left(\frac{a}{4}\right)^{2r^2}.
\end{equation}
Setting $\cvol:=2\log(4/a)$ proves Eq.~\eqref{eq:unit-volume-estimate}.

For the scaled estimate, we observe that $N_{C,\eta}$ is the dilation by $C L_\eta$ of the set $\{W:\normone{W}\le1,\ \norminf{W}\le1/(4r)\}$, whereas $\{W:\normone{W}\le L_\eta\}$ is the dilation by $L_\eta$ of $B_1^r$. 
Since the real dimension is $2r^2$, we have
\begin{equation}
\frac{\vol(N_{C,\eta})}{\vol\{W:\normone{W}\le L_\eta\}}=C^{2r^2}\frac{\vol\{W:\normone{W}\le1,\ \norminf{W}\le1/(4r)\}}{\vol(B_1^r)}\ge C^{2r^2}e^{-\cvol r^2}.
\end{equation}
This proves the claim.
\end{proof}

\section{Proof of the transcript-wise adaptive tester tilt lemma}\label{app:tilt-proof}

\begin{proof}[Proof of Lemma~\ref{lem:tilt}]
Fix $X_0\in G$, $C>1$, and a transcript $z=(y_1,\ldots,y_T)$ with positive probability under $\Phi_{X_0}$. 
For readability, write $\mathsf{T}_t=\mathsf{T}_t(z)$, and define
\begin{equation}
q_t:=\Tr(\mathsf{T}_tJ_{X_0}),\qquad
a_t(W):=\frac{\Tr(\mathsf{T}_t\Delta_W)}{q_t},
\end{equation}
where $\Delta_W:=J_{X_0+W}-J_{X_0}$.
Since $P_{X_0}(z)>0$, every single-step probability along the transcript is positive, and hence $q_t>0$. 
Moreover,
\begin{equation}
\Lambda_W(z)=\prod_{t=1}^T(1+a_t(W)).
\label{eq:Lambda-product-at}
\end{equation}

We first show that $a_t(W)$ is uniformly small. 
Since $X_0\in G$, Lemma~\ref{lem:valid-family} gives $J_{X_0}\ge\tfrac{1-3\sigma}{\D}I_{AB}$.
Therefore, we have
\begin{equation}
q_t=\Tr(\mathsf{T}_tJ_{X_0})\ge\frac{1-3\sigma}{\D}\Tr(\mathsf{T}_t).
\end{equation}
Since $\mathsf{T}_t\ge0$, we observe that $|\Tr(\mathsf{T}_t\Delta_W)|\le \norminf{\Delta_W}\Tr(\mathsf{T}_t)$, and that, for $W\in N_{C,\eta}$,
\begin{equation}
\norminf{\Delta_W}=\frac{\sigma}{\D}\norminf{E(W)}=\frac{\sigma}{\D}\norminf{W} \le\frac{\sigma}{\D}\frac{CL_\eta}{4r}=\frac{C\eta}{8r}\le\frac{3C\eta}{8\D},
\end{equation}
where we used $r\ge\D/3$. 
Hence, we obtain
\begin{equation}
|a_t(W)|\le\frac{3C\eta}{8(1-3\sigma)}\le\frac{1}{10}.
\label{eq:at-small}
\end{equation}

Next we estimate the second moment of $a_t(W)$ when $W$ is uniform on $N_{C,\eta}$. 
Let $P_0$ and $P_1$ be the projections onto $K_0$ and $K_1$, respectively, and set $B_t:=P_0\mathsf{T}_tP_1$, which is viewed as an operator from $K_1$ to $K_0$. Relative to $K_0\oplus K_1$, one obtain
\begin{equation}
E(W)=\begin{pmatrix}
0 & W\\
W^\dagger & 0
\end{pmatrix}.
\end{equation}
Since $\mathsf{T}_t$ is Hermitian,
\begin{equation}
\Tr(\mathsf{T}_tE(W))=2\re\,\Tr(B_t^\dagger W),\qquad
\bigl(\Tr(\mathsf{T}_tE(W))\bigr)^2\le 4\left|\Tr(B_t^\dagger W)\right|^2.
\label{eq:tester-offdiagonal-inner-product}
\end{equation}

The set $N_{C,\eta}$ is invariant under $W\mapsto UWV$ for all unitaries $U,V\in U(r)$ and under phase multiplication $W\mapsto e^{i\theta}W$.
Thus, if $W$ is uniform on $N_{C,\eta}$, its covariance is isotropic.
There exists $\alpha\ge0$ such that
\begin{equation}
\E_W[W_{ij}\overline{W_{kl}}]=\alpha\,\delta_{ik}\delta_{jl}.
\end{equation}
Taking traces gives
\begin{equation}
\alpha=\frac{\E_W\normtwo{W}^2}{r^2}.
\end{equation}
For every $W\in N_{C,\eta}$, note that $\normtwo{W}^2\le\normone{W}\norminf{W}\le(CL_\eta)\frac{CL_\eta}{4r}=\frac{C^2L_\eta^2}{4r}$.
Hence,
\begin{equation}
\alpha\le\frac{C^2L_\eta^2}{4r^3}.
\end{equation}
For every fixed $B\in\C^{r\times r}$,
\begin{equation}
\E_W|\Tr(B^\dagger W)|^2=\alpha\normtwo{B}^2.
\end{equation}
Using Eq.~\eqref{eq:tester-offdiagonal-inner-product} with $B=B_t$ and the projection contraction $\normtwo{B_t}\le\normtwo{\mathsf{T}_t}$, we obtain
\begin{equation}
\E_W\bigl(\Tr(\mathsf{T}_tE(W))\bigr)^2\le\frac{C^2L_\eta^2}{r^3}\normtwo{\mathsf{T}_t}^2.
\end{equation}
Since $\Delta_W=(\sigma/\D)E(W)$ and
$L_\eta=\eta\D/(2\sigma)$, we have
\begin{equation}
\E_W\bigl(\Tr(\mathsf{T}_t\Delta_W)\bigr)^2\le\frac{C^2\eta^2}{4r^3}\normtwo{\mathsf{T}_t}^2.
\end{equation}
Dividing by $q_t^2$ and using $q_t^2\ge\tfrac{(1-3\sigma)^2}{\D^2}(\Tr \mathsf{T}_t)^2$, we get
\begin{equation}
\E_W a_t(W)^2\le\frac{C^2\eta^2\D^2}{4(1-3\sigma)^2r^3}\frac{\normtwo{\mathsf{T}_t}^2}{(\Tr \mathsf{T}_t)^2}.
\end{equation}
Because $\mathsf{T}_t\ge0$ and thus $\normtwo{\mathsf{T}_t}^2\le(\Tr \mathsf{T}_t)^2$, using $r\ge\D/3$, we conclude that
\begin{equation}
\E_W a_t(W)^2\le\frac{\Ktilt}{2}\frac{C^2\eta^2}{\D},\qquad
\Ktilt:=\frac{27}{2(1-3\sigma)^2}.
\label{eq:at-second-moment}
\end{equation}

The set $N_{C,\eta}$ is symmetric under $W\mapsto -W$, while $a_t(W)$ is linear in $W$. 
Therefore
\begin{equation}
\E_W a_t(W)=0.
\label{eq:at-mean-zero}
\end{equation}
By Eq.~\eqref{eq:at-small} and the elementary inequality $\log(1+u)\ge u-2u^2$ for $|u|\le1/10$, we have
\begin{equation}
\E_W\log(1+a_t(W))\ge-2\E_Wa_t(W)^2\ge-\Ktilt\frac{C^2\eta^2}{\D}.
\end{equation}
Finally, Jensen's inequality gives
\begin{equation}
\log\E_W\Lambda_W(z)\ge\E_W\log\Lambda_W(z)=\sum_{t=1}^T\E_W\log(1+a_t(W))\ge-\Ktilt\frac{C^2\eta^2T}{\D},
\end{equation}
and thus
\begin{equation}
\E_W\Lambda_W(z)\ge\exp\left(-\Ktilt\frac{C^2\eta^2T}{\D}\right),
\end{equation}
which proves the lemma.
\end{proof}

\section{Proofs for posterior anti-concentration}\label{app:posterior-anti-concentration-proof}

\begin{proof}[Proof of Lemma~\ref{lem:common-support}]
Fix a history $h$ and an outcome $y$. 
Let $\mathsf{T}_{h,y}\ge0$ be the tester element associated with that history and outcome. 
For any $X\in S$,
\begin{equation}
\Prob_X[Y_t=y\mid h]=\din\,\Tr(\mathsf{T}_{h,y}J_X).
\end{equation}
By Lemma~\ref{lem:valid-family}, we have $J_X\succeq\tfrac{1-4\sigma}{\D}I_{AB}\succ0$, and thus
\begin{equation}
\Tr(\mathsf{T}_{h,y}J_X)>0\qquad\Longleftrightarrow\qquad \mathsf{T}_{h,y}\ne0.
\end{equation}
Indeed, if $\mathsf{T}_{h,y}=0$, the trace is zero. 
If $\mathsf{T}_{h,y}\ne0$, then $\Tr \mathsf{T}_{h,y}>0$, and
\begin{equation}
\Tr(\mathsf{T}_{h,y}J_X)\ge\frac{1-4\sigma}{\D}\Tr \mathsf{T}_{h,y}>0.
\end{equation}
Thus, whether a single-step conditional probability is positive is independent of $X\in S$.
For a full transcript $z=(y_1,\ldots,y_T)$,
\begin{equation}
P_X(z)=\prod_{t=1}^T\Prob_X[Y_t=y_t\mid y_{<t}].
\end{equation}
Each factor is positive for $X$ if and only if it is positive for $X'$. 
Hence $P_X(z)>0$ if and only if $P_{X'}(z)>0$.
\end{proof}

\begin{proof}[Proof of Lemma~\ref{lem:anti}]
We first choose the constants. 
Let $\cvol$ and $r_{\mathrm{vol}}$ be the constants in Lemma~\ref{lem:volume}, and let $\Cdens=4$ and $r_\mathrm{gin}$ be the constants in Lemma~\ref{lem:prior}. 
Choose $C>1$ large enough that
\begin{equation}
a_C:=\frac{2}{9}\log C-\frac{\cvol}{4}-(\Cdens+1)>0.
\end{equation}
Set
\begin{equation}
\kappa:=\frac{a_C}{2},\qquad
\cac:=\frac{\kappa}{\Ktilt C^2},\qquad
\eta_\star:=\min\left\{\frac{8\sigma}{3C},\frac{8(1-3\sigma)}{30C}\right\},\qquad
r_\star:=\max\{r_{\mathrm{gin}},r_{\mathrm{vol}}\}.
\end{equation}
Fix $r\ge r_\star$, $X_0\in G$, and $0<\eta\le\eta_\star$. 
Let $z$ be a transcript in the common support. 
We define
\begin{equation}
I_{\mathrm{small}}(z):=\int_{B_\eta(X_0)}R_X(z)\,dX .
\end{equation}
Using Tonelli's theorem and $\E_{Z\sim P_{X_0}}R_X(Z)=1$,
\begin{equation}
\E_{Z\sim P_{X_0}} I_{\mathrm{small}}(Z)=\int_{B_\eta(X_0)}\E_{Z\sim P_{X_0}}R_X(Z)\,dX=\vol(B_\eta(X_0))\le\vol(\widetilde B_\eta(X_0)).
\end{equation}
Therefore, by Markov's inequality,
\begin{equation}
P_{X_0}\left[I_{\mathrm{small}}(Z)>e^{\D^2}\vol(\widetilde B_\eta(X_0))\right]\le e^{-\D^2}.
\label{eq:anti-markov-bad}
\end{equation}
We work on the complementary event, where
\begin{equation}
I_{\mathrm{small}}(z)\le e^{\D^2}\vol(\widetilde B_\eta(X_0)).
\label{eq:anti-markov-good}
\end{equation}

Since $\eta\le\eta_\star$, the assumptions of Lemma~\ref{lem:support} and Lemma~\ref{lem:tilt} are satisfied. 
Thus
\begin{equation}
X_0+N_{C,\eta}\subseteq B_{C\eta}(X_0).
\end{equation}
For every transcript $z$ in the common support, translation invariance of Lebesgue measure gives
\begin{equation}
\int_{X_0+N_{C,\eta}}R_X(z)\,dX=\vol(N_{C,\eta})\,\E_{W\sim\mathrm{Unif}(N_{C,\eta})}R_{X_0+W}(z).
\end{equation}
Along the fixed transcript $z$, the adaptive tester elements are deterministic. 
Hence, for $W\in N_{C,\eta}$, we have
\begin{equation}
R_{X_0+W}(z)=\prod_{t=1}^T\frac{\Tr(\mathsf{T}_t(z)J_{X_0+W})}{\Tr(\mathsf{T}_t(z)J_{X_0})}=\Lambda_W(z).
\end{equation}
By Lemma~\ref{lem:tilt}, we have
\begin{equation}
\int_{B_{C\eta}(X_0)}R_X(z)\,dX\ge\int_{X_0+N_{C,\eta}}R_X(z)\,dX\ge\vol(N_{C,\eta})\exp\left(-\Ktilt\frac{C^2\eta^2T}{\D}\right).
\label{eq:anti-large-integral-lower}
\end{equation}

Let
\begin{equation}
m:=\inf_{X\in S}f_\mu(X),
\qquad
M:=\sup_{X\in S}f_\mu(X),
\end{equation}
by Lemma~\ref{lem:prior}, we have $\tfrac{M}{m}\le e^{\Cdens\D^2}$.
If $\nu_z(B_\eta(X_0))=0$, then the desired conclusion holds immediately.
Otherwise, on the event in Eq.~\eqref{eq:anti-markov-good},
\begin{equation}
\frac{\nu_z(B_{C\eta}(X_0))}{\nu_z(B_\eta(X_0))}\ge\frac{m}{M}\frac{\int_{B_{C\eta}(X_0)}R_X(z)\,dX}{\int_{B_\eta(X_0)}R_X(z)\,dX}\ge e^{-\Cdens\D^2}\frac{\vol(N_{C,\eta})}{e^{\D^2}\vol(\widetilde B_\eta(X_0))}\exp\left(-\Ktilt\frac{C^2\eta^2T}{\D}\right).
\end{equation}
Since
\begin{equation}
\vol(\widetilde B_\eta(X_0))=\vol\{W:\normone{W}\le L_\eta\},
\end{equation}
the volume estimate in Lemma~\ref{lem:volume} gives
\begin{equation}
\frac{\nu_z(B_{C\eta}(X_0))}{\nu_z(B_\eta(X_0))}\ge\exp\left(2r^2\log C-\cvol r^2-(\Cdens+1)\D^2-\Ktilt\frac{C^2\eta^2T}{\D}\right).
\end{equation}
Using $\D/3\le r\le \D/2$, we have
\begin{equation}
2r^2\log C-\cvol r^2-(\Cdens+1)\D^2\ge a_C\D^2.
\end{equation}
If $T\le \cac\tfrac{\D^3}{\eta^2}$, then
\begin{equation}
\Ktilt\frac{C^2\eta^2T}{\D}\le\kappa\D^2.
\end{equation}
Therefore, on the event in Eq.~\eqref{eq:anti-markov-good},
\begin{equation}
\frac{\nu_z(B_{C\eta}(X_0))}{\nu_z(B_\eta(X_0))}\ge e^{\kappa\D^2}.
\end{equation}
Since $\nu_z(B_{C\eta}(X_0))\le1$, it follows that $\nu_z(B_\eta(X_0))\le e^{-\kappa\D^2}$.
The good event in Eq.~\eqref{eq:anti-markov-good} has $P_{X_0}$-probability at
least $1-e^{-\D^2}$ by Eq.~\eqref{eq:anti-markov-bad}. This proves the lemma.
\end{proof}

\section{General outcome spaces}
\label{a:borel}

The main text is written for discrete outcome spaces in order to avoid measure-theoretic notation. 
We now record the standard-Borel formulation and explain why all proofs above extend without change.

Fix one round of an incoherent protocol at a history $h$. 
Let $(\calY_h,\Sigma_h)$ be a standard Borel outcome space, $\rho_h\in\calD(R_h\otimes A)$ be the chosen input state, and let $M_h:\Sigma_h\to\calL(R_h\otimes B)$ be a POVM. 
By the same purification-and-vectorization argument as in Lemma~\ref{lem:tester}, there is a positive operator-valued measure and a state $\tau_h\in\calD(A)$ such that for every channel $\Phi$, we have
\begin{equation}
\mathsf{T}_h:\Sigma_h\to\calL(A\otimes B),\qquad
\mathsf{T}_h(\calY_h)=\tau_h^\top\otimes I_B,\qquad
\Pr_\Phi[Y_t\in E\mid h]=\din\,\Tr\left(\mathsf{T}_h(E)J_\Phi\right),\quad E\in\Sigma_h .
\end{equation}
Since $A\otimes B$ is finite-dimensional, the scalar measure $\lambda_h(E):=\Tr \mathsf{T}_h(E)$ dominates the operator-valued measure $\mathsf{T}_h$. 
Hence there exists a $\lambda_h$-measurable positive-operator-valued density $y\mapsto \mathsf{T}_h(y)$ such that
\begin{equation}
    \mathsf{T}_h(E)=\int_E \mathsf{T}_h(y)\,d\lambda_h(y).
\end{equation}
Consequently the conditional distribution of $Y_t$, given $h$, has density
\begin{equation}
p_X(y\mid h)=\din\,\Tr\left(\mathsf{T}_h(y)J_X\right)
\end{equation}
with respect to $\lambda_h$, for every channel in the hard family.

The full-rank bound in Lemma~\ref{lem:valid-family} implies common conditional support. 
Indeed, for every $X\in S$, we have $J_X\succeq \tfrac{1-4\sigma}{\D}I_{AB}\succ0$, and therefore
\begin{equation}
p_X(y\mid h)>0\qquad\Longleftrightarrow\qquad
\mathsf{T}_h(y)\ne0.
\end{equation}
Thus the conditional support is independent of $X$. We note that the adaptive protocol induces mutually absolutely continuous transcript laws on the standard-Borel transcript space.
For a reference point $X_0\in S$, the likelihood ratio is, for $P_{X_0}$-almost every transcript $z=(y_1,\ldots,y_T)$,
\begin{equation}
\frac{dP_X}{dP_{X_0}}(z)=\prod_{t=1}^T\frac{\Tr\left(\mathsf{T}_t(z)J_X\right)}{\Tr\left(\mathsf{T}_t(z)J_{X_0}\right)},
\end{equation}
where $\mathsf{T}_t(z)$ denotes the density of the tester at the outcome observed at round $t$, along the history determined by $z$.

All arguments in Sections~\ref{sec:transcript-wise-tilt} and \ref{sec:posterior-anti-concentration} use only this product likelihood ratio, common support, positivity of the tester densities, and Markov inequalities. 
These ingredients hold exactly as above.
Therefore, Theorem~\ref{thm:main} and its corollaries extend from discrete outcome spaces to arbitrary standard Borel outcome spaces.

\section{Proof of the incoherent upper bound}
\label{app:log-free-upper-bound}

In this appendix, we prove Theorem~\ref{thm:log-free-upper}, removing the logarithmic factors appearing in the previously established upper bounds for incoherent channel tomography~\cite{kuengfasttkomogr,oufkir2023sample}. Our proof adapts the technique of Ref.~\cite{kuengfasttkomogr}, which gives an optimal upper bound without logarithmic overhead for incoherent state tomography, to the setting of quantum channels.
Throughout, for a unit vector $u$, write $P_u:=\proj{u}$. 
Let $\mu_d$ denote normalized Haar measure on the unit sphere of $\C^d$.
\subsection{The estimator}
\label{app:upper-estimator}

Fix an integer $T\ge1$. 
For each $t\in\{1,\ldots,T\}$, independently perform the following experiment.
\begin{enumerate}[label=(\roman*)]
\item Sample a Haar-random unit vector $v_t\in A$.
\item Prepare $P_{v_t}$, apply $\Phi$ once, and obtain the output state $\Phi(P_{v_t})$.
\item Independently sample $U_t$ from Haar measure on $U(\dout)$ and measure $\Phi(P_{v_t})$ in the orthonormal basis $\{U_t\ket{j}\}_{j=1}^{\dout}$. 
Let $I_t$ be the observed outcome and set $w_t:=U_t\ket{I_t}$.
\item Form the Hermitian random operator
\begin{equation}
X_t:=\left((\din+1)P_{v_t}^{\top}-I_A\right)\otimes\left((\dout+1)P_{w_t}-I_B\right).
\label{eq:upper-Xt-definition}
\end{equation}
\end{enumerate}
Define the linear estimator
\begin{equation}
L_T:=\frac{1}{T}\sum_{t=1}^T X_t.
\label{eq:upper-LT-definition}
\end{equation}
Let
\begin{equation}
\mathfrak C_{A\to B}:=\left\{K\in\calL(A\otimes B):K\ge0,\ \Tr_BK=\frac{I_A}{\din}\right\}
\label{eq:upper-choi-feasible-set}
\end{equation}
be the set of normalized Choi operators of channels from $A$ to $B$, and choose
\begin{equation}
\widetilde J\in\argmin_{K\in\mathfrak C_{A\to B}}\norminf{K-L_T}.
\label{eq:upper-projection-definition}
\end{equation}
Finally, return the unique channel $\widetilde\Phi$ whose normalized Choi operator is $\widetilde J$.

The minimizer in Eq.~\eqref{eq:upper-projection-definition} exists. 
Indeed, $\mathfrak C_{A\to B}$ is closed, and every $K\in\mathfrak C_{A\to B}$ satisfies $K\ge0$ and $\Tr K=1$, hence $\norminf{K}\le1$. 
Thus $\mathfrak C_{A\to B}$ is compact in finite dimension, and the objective in Eq.~\eqref{eq:upper-projection-definition} is continuous. 
All vectors $v_t$ and bases $U_t$ may be sampled before the experiment begins. 
Each $X_t$ uses exactly one query of $\Phi$, each channel output is measured separately, and no ancillary system is used. 
The protocol is therefore non-adaptive, incoherent, and ancilla-free.

\subsection{Haar moments and random-basis reconstruction}
\label{app:upper-haar-reconstruction}

\begin{lemma}[First three Haar moments~\cite{harrow2013churchsymmetricsubspace,Mele_2024}]
\label{lem:upper-haar-moments}
Let $u\sim\mu_d$, and let $\Pi_{\mathrm{sym}}^{(k)}$ be the orthogonal projector onto the symmetric subspace of $(\C^d)^{\otimes k}$. 
For $k\in\{1,2,3\}$,
\begin{equation}
\E_{u\sim\mu_d}\left[P_u^{\otimes k}\right]
=\frac{\Pi_{\mathrm{sym}}^{(k)}}{\binom{d+k-1}{k}}.
\label{eq:upper-general-haar-moment}
\end{equation}
In particular, if $F$ denotes the swap operator on $\C^d\otimes\C^d$, then
\begin{equation}
\E P_u=\frac{I_d}{d},\qquad
\E P_u^{\otimes2}=\frac{I+F}{d(d+1)},\qquad
\E P_u^{\otimes3}=\frac{\Pi_{\mathrm{sym}}^{(3)}}{\binom{d+2}{3}}.
\label{eq:upper-first-three-haar-moments}
\end{equation}
\end{lemma}

The following lemma is standard and has previously been used in the contexts of quantum state tomography and classical shadows~\cite{kuengfasttkomogr,huangclassical}. For completeness, we provide a self-contained proof with explicit constants.
\begin{lemma}[Random-basis reconstruction]
\label{lem:upper-random-basis-reconstruction}
Let $\rho\in\calD(\C^d)$. 
Sample $U$ from Haar measure on $U(d)$, measure $\rho$ in the basis $\{U\ket{j}\}_{j=1}^d$, let $I$ be the outcome, and set $w:=U\ket{I}$. 
Then, for every bounded measurable function $f$ on the unit sphere,
\begin{equation}
\E f(w)=d\int f(w)\bra{w}\rho\ket{w}\,d\mu_d(w).
\label{eq:upper-outcome-law}
\end{equation}
Define
\begin{equation}
Y_d(w):=(d+1)P_w-I_d.
\label{eq:upper-Yd-definition}
\end{equation}
Then
\begin{equation}
\E Y_d(w)=\rho.
\label{eq:upper-local-unbiasedness}
\end{equation}
Moreover, for every $Q\ge0$,
\begin{equation}
\E\left[\Tr\left(QY_d(w)\right)^2\right]\le14\left(\Tr Q\right)^2.
\label{eq:upper-local-second-moment}
\end{equation}
\end{lemma}

\begin{proof}
Conditioned on $U$, Born's rule gives
\begin{equation}
\E[f(w)\mid U]=\sum_{j=1}^d\bra{j}U^\dagger\rho U\ket{j}\,f(U\ket{j}).
\end{equation}
For each fixed $j$, the vector $U\ket{j}$ is Haar-distributed. 
Taking expectation over $U$ therefore gives
\begin{equation}
\E f(w)=d\int f(w)\bra{w}\rho\ket{w}\,d\mu_d(w),
\end{equation}
which proves Eq.~\eqref{eq:upper-outcome-law}.

Using Eq.~\eqref{eq:upper-outcome-law} and the second Haar moment,
\begin{equation}
\begin{aligned}
\E P_w
&=d\int \Tr(\rho P_w)P_w\,d\mu_d(w)\\
&=d\,\Tr_1\left[(\rho\otimes I_d)\frac{I+F}{d(d+1)}\right]\\
&=\frac{I_d+\rho}{d+1}.
\end{aligned}
\label{eq:upper-EPw}
\end{equation}
The last equality uses $\Tr_1[(\rho\otimes I_d)I]=I_d$ and $\Tr_1[(\rho\otimes I_d)F]=\rho$. 
Equation~\eqref{eq:upper-local-unbiasedness} follows immediately from Eqs.~\eqref{eq:upper-Yd-definition} and \eqref{eq:upper-EPw}.

For the second moment, set $q:=\Tr Q$. 
By Eq.~\eqref{eq:upper-outcome-law} and the inequality $(x-y)^2\le2x^2+2y^2$,
\begin{equation}
\begin{aligned}
\E\left[\Tr\left(QY_d(w)\right)^2\right]
&=d\int \Tr(\rho P_w)\left((d+1)\Tr(QP_w)-q\right)^2d\mu_d(w)\\
&\le2d(d+1)^2\int \Tr(\rho P_w)\Tr(QP_w)^2d\mu_d(w)+2q^2.
\end{aligned}
\label{eq:upper-local-second-prebound}
\end{equation}
The third Haar moment gives
\begin{equation}
\int \Tr(\rho P_w)\Tr(QP_w)^2d\mu_d(w)
=\frac{\Tr\left[(\rho\otimes Q\otimes Q)\Pi_{\mathrm{sym}}^{(3)}\right]}{\binom{d+2}{3}}.
\label{eq:upper-third-moment-contraction}
\end{equation}
Since $\rho\otimes Q\otimes Q\ge0$ and $0\le\Pi_{\mathrm{sym}}^{(3)}\le I$, we have
\begin{equation}
\Tr\left[(\rho\otimes Q\otimes Q)\Pi_{\mathrm{sym}}^{(3)}\right]
\le\Tr(\rho\otimes Q\otimes Q)=q^2.
\label{eq:upper-third-moment-upper}
\end{equation}
Substituting Eqs.~\eqref{eq:upper-third-moment-contraction} and \eqref{eq:upper-third-moment-upper} into Eq.~\eqref{eq:upper-local-second-prebound} yields
\begin{equation}
\begin{aligned}
\E\left[\Tr\left(QY_d(w)\right)^2\right]
&\le2d(d+1)^2\frac{6q^2}{d(d+1)(d+2)}+2q^2\\
&=12\frac{d+1}{d+2}q^2+2q^2\\
&\le14q^2.
\end{aligned}
\end{equation}
This proves Eq.~\eqref{eq:upper-local-second-moment}.
\end{proof}

\subsection{Unbiasedness of the Choi estimator}
\label{app:upper-unbiasedness}

\begin{lemma}[Choi frame identity]
\label{lem:upper-choi-frame}
For $v\sim\mu_{\din}$,
\begin{equation}
(\din+1)\E_v\left[P_v^\top\otimes\Phi(P_v)\right]
=J_\Phi+I_A\otimes\Phi\left(\frac{I_A}{\din}\right).
\label{eq:upper-choi-frame}
\end{equation}
\end{lemma}

\begin{proof}
Let $F$ be the swap operator on $A\otimes A$. 
Taking the transpose on the first tensor factor of the second Haar-moment identity gives
\begin{equation}
\E_v\left[P_v^\top\otimes P_v\right]
=\frac{I_A\otimes I_A+F^{\top_1}}{\din(\din+1)}.
\label{eq:upper-partial-transpose-haar}
\end{equation}
A direct expansion of the swap operator shows that
\begin{equation}
F^{\top_1}=\din\proj{\Omega}.
\label{eq:upper-swap-partial-transpose}
\end{equation}
Applying $\id_A\otimes\Phi$ to Eq.~\eqref{eq:upper-partial-transpose-haar}, using Eq.~\eqref{eq:upper-swap-partial-transpose} and the definition of $J_\Phi$, gives
\begin{equation}
\E_v\left[P_v^\top\otimes\Phi(P_v)\right]
=\frac{I_A\otimes\Phi(I_A)+\din J_\Phi}{\din(\din+1)}.
\end{equation}
Multiplying by $\din+1$ proves Eq.~\eqref{eq:upper-choi-frame}.
\end{proof}

\begin{lemma}[Unbiasedness]
\label{lem:upper-unbiasedness}
For every $t$,
\begin{equation}
\E X_t=J_\Phi.
\label{eq:upper-Xt-unbiased}
\end{equation}
Consequently, $\E L_T=J_\Phi$.
\end{lemma}

\begin{proof}
Condition on $v_t=v$. 
The measured output state is $\rho_v:=\Phi(P_v)$. 
By Lemma~\ref{lem:upper-random-basis-reconstruction},
\begin{equation}
\E\left[(\dout+1)P_{w_t}-I_B\mid v_t=v\right]=\Phi(P_v).
\label{eq:upper-conditional-output-reconstruction}
\end{equation}
Using the tower property, Eq.~\eqref{eq:upper-conditional-output-reconstruction}, and $\E_vP_v=I_A/\din$, we obtain
\begin{equation}
\begin{aligned}
\E X_t
&=\E_v\left[\left((\din+1)P_v^\top-I_A\right)\otimes\Phi(P_v)\right]\\
&=(\din+1)\E_v\left[P_v^\top\otimes\Phi(P_v)\right]
-I_A\otimes\Phi\left(\frac{I_A}{\din}\right)\\
&=J_\Phi,
\end{aligned}
\end{equation}
where the final equality is Lemma~\ref{lem:upper-choi-frame}. 
Linearity gives $\E L_T=J_\Phi$.
\end{proof}

\subsection{Dimension-independent directional variance}
\label{app:upper-directional-variance}

\begin{lemma}[Directional variance and range]
\label{lem:upper-directional-variance}
For every unit vector $z\in A\otimes B$, define
\begin{equation}
G_z:=\bra{z}X_t\ket{z}.
\end{equation}
Then
\begin{equation}
\operatorname{Var}(G_z)\le140.
\label{eq:upper-directional-variance}
\end{equation}
Moreover,
\begin{equation}
\left|G_z-\bra{z}J_\Phi\ket{z}\right|\le2\D
\label{eq:upper-directional-range}
\end{equation}
almost surely.
\end{lemma}

\begin{proof}
Let
\begin{equation}
R_A:=\Tr_B\proj{z},\qquad R_B:=\Tr_A\proj{z}.
\label{eq:upper-reduced-states}
\end{equation}
Both $R_A$ and $R_B$ are density operators. 
For fixed $v$, define the positive operator on $B$
\begin{equation}
Q_v:=(\bra{\overline v}\otimes I_B)\proj{z}(\ket{\overline v}\otimes I_B),
\label{eq:upper-Qv-definition}
\end{equation}
and set
\begin{equation}
q_v:=\Tr Q_v=\bra{\overline v}R_A\ket{\overline v}.
\label{eq:upper-qv-definition}
\end{equation}
Writing $Y:=(\dout+1)P_{w_t}-I_B$ and using $P_v^\top=P_{\overline v}$, the definition of the partial trace gives
\begin{equation}
G_z=(\din+1)\Tr(Q_vY)-\Tr(R_BY).
\label{eq:upper-Gz-decomposition}
\end{equation}
Conditioned on $v$, the vector $w_t$ is obtained from the state $\Phi(P_v)$ by the random-basis procedure in Lemma~\ref{lem:upper-random-basis-reconstruction}. 
Applying Eq.~\eqref{eq:upper-local-second-moment} to $Q_v$ and to $R_B$, respectively, gives
\begin{equation}
\E\left[\Tr(Q_vY)^2\mid v\right]\le14q_v^2,
\qquad
\E\left[\Tr(R_BY)^2\mid v\right]\le14.
\label{eq:upper-two-local-moments}
\end{equation}
Using $(x-y)^2\le2x^2+2y^2$ in Eq.~\eqref{eq:upper-Gz-decomposition},
\begin{equation}
\E[G_z^2\mid v]\le28(\din+1)^2q_v^2+28.
\label{eq:upper-Gz-conditional-second}
\end{equation}
Because $\overline v$ is Haar-distributed whenever $v$ is Haar-distributed, the second Haar moment gives
\begin{equation}
\begin{aligned}
\E_vq_v^2
&=\E_v\bra{\overline v}R_A\ket{\overline v}^2\\
&=\frac{(\Tr R_A)^2+\Tr(R_A^2)}{\din(\din+1)}\\
&\le\frac{2}{\din(\din+1)}.
\end{aligned}
\label{eq:upper-qv-second-moment}
\end{equation}
Averaging Eq.~\eqref{eq:upper-Gz-conditional-second} over $v$ and using Eq.~\eqref{eq:upper-qv-second-moment},
\begin{equation}
\begin{aligned}
\E G_z^2
&\le28(\din+1)^2\frac{2}{\din(\din+1)}+28\\
&=56\frac{\din+1}{\din}+28\\
&\le140.
\end{aligned}
\label{eq:upper-Gz-second-moment}
\end{equation}
By Lemma~\ref{lem:upper-unbiasedness}, $\E G_z=\bra{z}J_\Phi\ket{z}$. 
Therefore,
\begin{equation}
\operatorname{Var}(G_z)=\E G_z^2-(\E G_z)^2\le\E G_z^2\le140,
\end{equation}
which proves Eq.~\eqref{eq:upper-directional-variance}.

For the range bound, $(\din+1)P_v^\top-I_A$ has operator norm $\din$, and $(\dout+1)P_w-I_B$ has operator norm $\dout$. 
Hence
\begin{equation}
\norminf{X_t}=\din\dout=\D.
\label{eq:upper-Xt-operator-norm}
\end{equation}
Since $J_\Phi$ is a density operator, $\norminf{J_\Phi}\le1$. 
Thus
\begin{equation}
\left|G_z-\bra{z}J_\Phi\ket{z}\right|
\le\norminf{X_t-J_\Phi}
\le\D+1
\le2\D,
\end{equation}
which proves Eq.~\eqref{eq:upper-directional-range}.
\end{proof}

\subsection{Scalar concentration and a covering net}
\label{app:upper-concentration}
 
We recall the following standard form of Bernstein's inequality; see, for
example, Ref.~\cite[Eq.~(2.10), p.~36]{boucheron2013concentration}.
For completeness, we include a self-contained proof that keeps track of
the constants.

\begin{lemma}[Scalar Bernstein inequality]
\label{lem:upper-scalar-bernstein}
Let $\xi_1,\ldots,\xi_T$ be independent mean-zero real random variables satisfying $|\xi_t|\le R$ almost surely and $\E\xi_t^2\le\sigma^2$ for every $t$. 
Then, for every $s>0$,
\begin{equation}
\Prob\left[\left|\frac{1}{T}\sum_{t=1}^T\xi_t\right|\ge s\right]
\le2\exp\left(-\frac{Ts^2}{2\left(\sigma^2+Rs/3\right)}\right).
\label{eq:upper-scalar-bernstein}
\end{equation}
\end{lemma}

\begin{proof}
Assume first that $\sigma^2>0$. 
For $0\le\lambda<3/R$ and every real $x$ with $|x|\le R$, the power-series expansion of the exponential and the bound $1/k!\le1/(2\cdot3^{k-2})$ for $k\ge2$ imply
\begin{equation}
e^{\lambda x}
\le1+\lambda x+\frac{\lambda^2x^2}{2(1-\lambda R/3)}.
\label{eq:upper-bernstein-elementary-mgf}
\end{equation}
Taking expectations and using $\E\xi_t=0$ gives
\begin{equation}
\E e^{\lambda\xi_t}
\le1+\frac{\lambda^2\E\xi_t^2}{2(1-\lambda R/3)}
\le\exp\left(\frac{\lambda^2\sigma^2}{2(1-\lambda R/3)}\right).
\label{eq:upper-bernstein-one-mgf}
\end{equation}
Independence and Markov's inequality therefore imply
\begin{equation}
\Prob\left[\sum_{t=1}^T\xi_t\ge Ts\right]
\le\exp\left(-\lambda Ts+\frac{T\lambda^2\sigma^2}{2(1-\lambda R/3)}\right).
\label{eq:upper-bernstein-chernoff}
\end{equation}
Choose
\begin{equation}
\lambda:=\frac{s}{\sigma^2+Rs/3}.
\label{eq:upper-bernstein-lambda}
\end{equation}
Then $0\le\lambda<3/R$ and
\begin{equation}
1-\frac{\lambda R}{3}=\frac{\sigma^2}{\sigma^2+Rs/3}.
\end{equation}
Substitution into Eq.~\eqref{eq:upper-bernstein-chernoff} yields
\begin{equation}
\Prob\left[\frac{1}{T}\sum_{t=1}^T\xi_t\ge s\right]
\le\exp\left(-\frac{Ts^2}{2(\sigma^2+Rs/3)}\right).
\end{equation}
Applying the same argument to $-\xi_t$ and taking a union bound proves Eq.~\eqref{eq:upper-scalar-bernstein}. 
If $\sigma^2=0$, then each $\xi_t=0$ almost surely, and the conclusion is immediate.
\end{proof}

\begin{lemma}[Constant-radius net]
\label{lem:upper-net}
There exists a $1/4$-net $\mathcal T$ of the unit sphere of $\C^\D$, in Euclidean norm, such that
\begin{equation}
|\mathcal T|\le9^{2\D}.
\label{eq:upper-net-size}
\end{equation}
For every Hermitian operator $H\in\calL(\C^\D)$,
\begin{equation}
\norminf{H}\le2\max_{z\in\mathcal T}|\bra{z}H\ket{z}|.
\label{eq:upper-net-operator-norm}
\end{equation}
\end{lemma}

\begin{proof}
Identify $\C^\D$ with $\mathbb R^{2\D}$ and let $\mathcal T$ be a maximal $1/4$-separated subset of the unit sphere. 
Maximality implies that $\mathcal T$ is a $1/4$-net. 
The Euclidean balls of radius $1/8$ centered at points of $\mathcal T$ are pairwise disjoint and are all contained in the ball of radius $9/8$. 
Comparing $2\D$-dimensional Euclidean volumes gives
\begin{equation}
|\mathcal T|\left(\frac18\right)^{2\D}\le\left(\frac98\right)^{2\D},
\end{equation}
which proves Eq.~\eqref{eq:upper-net-size}.

Let $H$ be Hermitian, and let $x$ be a unit eigenvector with $|\bra{x}H\ket{x}|=\norminf{H}$. 
Choose $z\in\mathcal T$ with $\normtwo{x-z}\le1/4$. 
Then
\begin{equation}
\begin{aligned}
|\bra{x}H\ket{x}-\bra{z}H\ket{z}|
&\le|\bra{x-z}H\ket{x}|+|\bra{z}H\ket{x-z}|\\
&\le2\normtwo{x-z}\norminf{H}\\
&\le\frac12\norminf{H}.
\end{aligned}
\end{equation}
Consequently, $|\bra{z}H\ket{z}|\ge\norminf{H}/2$, proving Eq.~\eqref{eq:upper-net-operator-norm}.
\end{proof}

\begin{proposition}[Operator-norm concentration]
\label{prop:upper-operator-concentration}
For every $0<\varepsilon\le1$,
\begin{equation}
\Prob\left[\norminf{L_T-J_\Phi}>\frac{\varepsilon}{2\D}\right]
\le2\exp\left(2\D\log9-\frac{T\varepsilon^2}{4496\D^2}\right).
\label{eq:upper-operator-tail}
\end{equation}
Consequently, for every $0<\delta<1$, if
\begin{equation}
T\ge\left\lceil\frac{4496\D^2}{\varepsilon^2}\left(2\D\log9+\log\frac{2}{\delta}\right)\right\rceil,
\label{eq:upper-explicit-sample-size}
\end{equation}
then
\begin{equation}
\Prob\left[\norminf{L_T-J_\Phi}\le\frac{\varepsilon}{2\D}\right]\ge1-\delta.
\label{eq:upper-good-event}
\end{equation}
\end{proposition}

\begin{proof}
Fix a unit vector $z\in A\otimes B$ and define
\begin{equation}
\xi_t(z):=\bra{z}(X_t-J_\Phi)\ket{z}.
\end{equation}
The variables $\xi_t(z)$ are independent and real. 
Lemma~\ref{lem:upper-unbiasedness} gives $\E\xi_t(z)=0$, and Lemma~\ref{lem:upper-directional-variance} gives
\begin{equation}
\E\xi_t(z)^2\le140,
\qquad
|\xi_t(z)|\le2\D.
\label{eq:upper-xi-hypotheses}
\end{equation}
Apply Lemma~\ref{lem:upper-scalar-bernstein} with $\sigma^2=140$, $R=2\D$, and $s=\varepsilon/(4\D)$. 
Since $0<\varepsilon\le1$,
\begin{equation}
2\left(\sigma^2+\frac{Rs}{3}\right)
=280+\frac{\varepsilon}{3}
\le281.
\label{eq:upper-bernstein-denominator}
\end{equation}
Therefore,
\begin{equation}
\Prob\left[\left|\bra{z}(L_T-J_\Phi)\ket{z}\right|\ge\frac{\varepsilon}{4\D}\right]
\le2\exp\left(-\frac{T\varepsilon^2}{4496\D^2}\right).
\label{eq:upper-fixed-direction-tail}
\end{equation}

Set $H:=L_T-J_\Phi$ and let $\mathcal T$ be the net from Lemma~\ref{lem:upper-net}. 
By Eq.~\eqref{eq:upper-net-operator-norm}, the event $\norminf{H}>\varepsilon/(2\D)$ implies that there exists $z\in\mathcal T$ such that $|\bra{z}H\ket{z}|>\varepsilon/(4\D)$. 
The union bound, Eqs.~\eqref{eq:upper-net-size} and \eqref{eq:upper-fixed-direction-tail}, give
\begin{equation}
\begin{aligned}
\Prob\left[\norminf{H}>\frac{\varepsilon}{2\D}\right]
&\le2|\mathcal T|\exp\left(-\frac{T\varepsilon^2}{4496\D^2}\right)\\
&\le2\exp\left(2\D\log9-\frac{T\varepsilon^2}{4496\D^2}\right).
\end{aligned}
\end{equation}
This proves Eq.~\eqref{eq:upper-operator-tail}. 
If Eq.~\eqref{eq:upper-explicit-sample-size} holds, the right-hand side of Eq.~\eqref{eq:upper-operator-tail} is at most $\delta$, proving Eq.~\eqref{eq:upper-good-event}.
\end{proof}

\subsection{Projection and conversion to diamond norm}
\label{app:upper-projection-diamond}

\begin{lemma}[Operator norm of the Choi difference controls diamond norm]
\label{lem:upper-choi-to-diamond}
Let $\Phi,\Psi:\calL(A)\to\calL(B)$ be quantum channels. 
Then
\begin{equation}
\normdiam{\Phi-\Psi}\le\din\dout\,\norminf{J_\Phi-J_\Psi}
=\D\norminf{J_\Phi-J_\Psi}.
\label{eq:upper-choi-to-diamond}
\end{equation}
\end{lemma}

\begin{proof}
Set $\Delta:=\Phi-\Psi$ and $J_\Delta:=J_\Phi-J_\Psi$. 
The map $\Delta$ is Hermiticity preserving. 
By the standard pure-state characterization of the diamond norm, with a reference system of dimension $\din$~\cite{Watrous2018},
\begin{equation}
\normdiam{\Delta}
=\max_{\substack{\ket{\psi}\in A\otimes A\\\normtwo{\psi}=1}}
\normone{(\id_A\otimes\Delta)(\proj{\psi})}.
\label{eq:upper-pure-state-diamond}
\end{equation}
Every unit vector $\ket{\psi}\in A\otimes A$ can be written as
\begin{equation}
\ket{\psi}=(M\otimes I_A)\ket{\Omega}
\label{eq:upper-filtered-max-entangled}
\end{equation}
for some $M\in\calL(A)$. 
The normalization of $\ket{\psi}$ implies
\begin{equation}
1=\bra{\Omega}(M^\dagger M\otimes I_A)\ket{\Omega}
=\frac{1}{\din}\Tr(M^\dagger M),
\label{eq:upper-M-normalization}
\end{equation}
so $\Tr(M^\dagger M)=\din$. 
Since $M$ acts on the untouched reference register,
\begin{equation}
(\id_A\otimes\Delta)(\proj{\psi})
=(M\otimes I_B)J_\Delta(M^\dagger\otimes I_B).
\label{eq:upper-filtered-choi}
\end{equation}
Let
\begin{equation}
J_\Delta=\sum_{k=1}^{\D}\lambda_k\proj{u_k}
\end{equation}
be a spectral decomposition with an orthonormal eigenbasis, including eigenvectors with zero eigenvalue. 
By the triangle inequality and positivity of each $(M\otimes I_B)\proj{u_k}(M^\dagger\otimes I_B)$,
\begin{equation}
\begin{aligned}
\normone{(M\otimes I_B)J_\Delta(M^\dagger\otimes I_B)}
&\le\sum_{k=1}^{\D}|\lambda_k|\,\bra{u_k}(M^\dagger M\otimes I_B)\ket{u_k}\\
&\le\norminf{J_\Delta}\sum_{k=1}^{\D}\bra{u_k}(M^\dagger M\otimes I_B)\ket{u_k}\\
&=\norminf{J_\Delta}\Tr(M^\dagger M\otimes I_B)\\
&=\din\dout\,\norminf{J_\Delta},
\end{aligned}
\end{equation}
where Eq.~\eqref{eq:upper-M-normalization} was used in the last step. 
Maximizing over $\ket{\psi}$ in Eq.~\eqref{eq:upper-pure-state-diamond} proves Eq.~\eqref{eq:upper-choi-to-diamond}.
\end{proof}

\begin{proof}[Proof of Theorem~\ref{thm:log-free-upper}]
Run the protocol of Appendix~\ref{app:upper-estimator} with $T$ satisfying Eq.~\eqref{eq:upper-explicit-sample-size}. 
By Proposition~\ref{prop:upper-operator-concentration}, with probability at least $1-\delta$,
\begin{equation}
\norminf{L_T-J_\Phi}\le\frac{\varepsilon}{2\D}.
\label{eq:upper-proof-good-event}
\end{equation}
On this event, $J_\Phi\in\mathfrak C_{A\to B}$ is feasible in Eq.~\eqref{eq:upper-projection-definition}. 
The optimality of $\widetilde J$ therefore implies
\begin{equation}
\norminf{\widetilde J-L_T}
\le\norminf{J_\Phi-L_T}
\le\frac{\varepsilon}{2\D}.
\label{eq:upper-projection-optimality}
\end{equation}
By the triangle inequality and Eqs.~\eqref{eq:upper-proof-good-event} and \eqref{eq:upper-projection-optimality},
\begin{equation}
\norminf{\widetilde J-J_\Phi}
\le\norminf{\widetilde J-L_T}+\norminf{L_T-J_\Phi}
\le\frac{\varepsilon}{\D}.
\label{eq:upper-projected-choi-error}
\end{equation}
Applying Lemma~\ref{lem:upper-choi-to-diamond} gives
\begin{equation}
\normdiam{\widetilde\Phi-\Phi}
\le\D\norminf{\widetilde J-J_\Phi}
\le\varepsilon.
\end{equation}
Thus Eq.~\eqref{eq:upper-bound-main-guarantee} holds.

It remains to verify the sample complexity. 
Let $T_0$ be the least integer satisfying Eq.~\eqref{eq:upper-explicit-sample-size}. 
Since $\log(2/\delta)=\log2+\log(1/\delta)$, $\D\ge1$, and $0<\varepsilon\le1$, we have
\begin{equation}
\begin{aligned}
T_0
&\le1+\frac{4496\D^2}{\varepsilon^2}\left(2\D\log9+\log2+\log(1/\delta)\right)\\
&\le\left(4496(2\log9+\log2)+1\right)\frac{\D^3}{\varepsilon^2}
+4496\frac{\D^2\log(1/\delta)}{\varepsilon^2}\\
&\le30000\,\frac{\D^3+\D^2\log(1/\delta)}{\varepsilon^2}.
\end{aligned}
\label{eq:upper-high-probability-complexity}
\end{equation}
The last inequality uses $4496(2\log9+\log2)+1<30000$ and $4496<30000$. 
Thus Theorem~\ref{thm:log-free-upper} holds with $C=30000$. 
For $\delta=1/3$, the choice
\begin{equation}
T=\left\lceil30000\,\frac{\D^3}{\varepsilon^2}\right\rceil
\label{eq:upper-constant-confidence-explicit}
\end{equation}
satisfies Eq.~\eqref{eq:upper-explicit-sample-size}, because $\D\ge1$ and
\begin{equation}
30000>4496\left(2\log9+\log6\right).
\end{equation}
Since $\D^3=\din^3\dout^3$, this proves Theorem~\ref{thm:log-free-upper}.
\end{proof}

\end{document}